\documentclass{vldb}
\usepackage{enumitem}
\usepackage{booktabs}
\usepackage{bbding}
\usepackage{graphicx}
\usepackage{balance}  
\usepackage{graphicx}
\usepackage{balance}
\usepackage[noend]{algorithmic}
\usepackage{subfigure}
\usepackage{multirow}
\usepackage{times}
\usepackage{color}
\usepackage{colortbl}
\usepackage{times}
\usepackage{amsmath}
\usepackage{epsfig}
\usepackage{amsfonts}
\usepackage[misc]{ifsym}
\usepackage{url}

\newcommand{\myskip}{\vspace{2pt}}

\newcommand{\tabincell}[2]{\begin{tabular}{@{}#1@{}}#2\end{tabular}}

\setenumerate[1]{itemsep=0pt,partopsep=0pt,parsep=2pt,topsep=0pt, leftmargin=*}
\setitemize[1]{itemsep=0pt,partopsep=0pt,parsep=2pt,topsep=0pt, leftmargin=*}

\definecolor{mygray}{gray}{.9}
\definecolor{myemph}{rgb}{1,0.73,0.74}

\newcommand{\G}{\mathcal{G}}
\newcommand{\A}{\mathbf{A}}
\newcommand{\HH}{\mathbf{I}}
\newcommand{\F}{\mathbf{F}}
\newcommand{\M}{\mathbf{M}}
\newcommand{\0}{\mathbf{0}}
\newcommand{\Y}{\mathbf{Y}}

\newcommand{\figwidth}{3.3in}

\newcommand{\algswidth}{\columnwidth}
\newcommand{\et}{
\setlength{\abovedisplayskip}{-0.5pt}
\setlength{\belowdisplayskip}{-0.5pt}
\setlength{\belowdisplayshortskip}{-0.5pt}
\setlength{\abovedisplayshortskip}{-0.5pt}
}

\newtheorem{lemma}{Lemma}

\vldbTitle{QuickIM: Efficient, Accurate and Robust Influence Maximization Algorithm
on Billion-Scale Networks}
\vldbAuthors{Rong Zhu, Zhaonian Zou, Yue Han, Sheng Yang, Jianzhong Li}
\vldbDOI{https://doi.org/TBD}

\begin{document}

\title{QuickIM: Efficient, Accurate and Robust Influence Maximization Algorithm
on Billion-Scale Networks}

\numberofauthors{1}
\author{
\alignauthor
Rong Zhu, Zhaonian Zou, Yue Han, Sheng Yang, and Jianzhong Li\\
       \affaddr{Harbin Institute of Technology, Harbin, China}\\
       \email{\{rzhu, znzou, yuehan, yangsheng, lijzh\}@hit.edu.cn}
}

\maketitle

\begin{abstract}
The Influence Maximization (IM) problem aims at finding $k$ seed vertices in a network, starting from which influence can be spread in the network to the maximum extent.
In this paper, we propose \textsf{QuickIM}, the first versatile IM algorithm that attains all the desirable properties of a practically applicable IM algorithm at the same time, namely high time efficiency, good result quality, low memory footprint, and high robustness.
On real-world social networks, \textsf{QuickIM} achieves the $\Omega(n + m)$ lower bound on time complexity and $\Omega(n)$ space complexity, where $n$ and $m$ are the number of vertices and edges in the network, respectively.
Our experimental evaluation verifies the superiority of \textsf{QuickIM}.
Firstly, \textsf{QuickIM} runs 1--3 orders of magnitude faster than the state-of-the-art IM algorithms.
Secondly, except \textsf{EasyIM}, \textsf{QuickIM} requires 1--2 orders of magnitude less memory than the state-of-the-art algorithms.
Thirdly, \textsf{QuickIM} always produces as good quality results as the state-of-the-art algorithms.
Lastly, the time and the memory performance of \textsf{QuickIM} is independent of influence probabilities.
On the largest network used in the experiments that contains more than 3.6 billion edges, \textsf{QuickIM} is able to find hundreds of influential seeds in less than 4 minutes, while all the state-of-the-art algorithms fail to terminate in an hour.
\end{abstract}

\vspace{-0.5em}

\section{Introduction}
\label{Sec: IntDuc}

Given a social network $G$, a budget $k \in \mathbb{N}$ and an influence diffusion model, the \emph{Influence Maximization (IM)} problem identifies $k$ vertices in $G$, called ``seeds'', that can influence the largest number of vertices according to the diffusion model. The IM problem paves the way for many real-world applications, e.g., viral marketing~\cite{Aslay2014Viral, Richardson2002Mining}, recommendation~\cite{Chaoji2012Recommendations}, rumor blocking~\cite{Tong2017An}, epidemic prevention~\cite{Tong2015Adaptive}, and so on.

Kempe et al.~\cite{Kempe2003Maximizing} first formalize the well-known \emph{Independent Cascade (IC)} model. Although the IM problem is NP-hard under the IC model, it can be solved by a simple greedy algorithm with an approximation ratio of $1 - 1/e - \epsilon$, where $e$ is the base of natural logarithms, and $\epsilon \in (0, 1)$ is a small number. After that, a vast number of algorithms have been developed to efficiently and accurately solve the IM problem under the IC model.
These algorithms can be categorized into three groups: \emph{the simulation-based algorithms}~\cite{Cheng2013StaticGreedy, Goyal2011CELF, Kempe2003Maximizing, Leskovec2007Cost, Ohsaka2014Fast}, \emph{the reverse sampling algorithms}~\cite{Borgs2012Maximizing, Huang2017Revisiting, Nguyen2017Importance, Nguyen2016Stop, Ohsaka2017Coarsening, Tang2015Influence, Tang2014Influence, Wang2017Bring} and \emph{the score estimation algorithms}~\cite{Chen2010Scalable, Galhotra2016Holistic, Goyal2012SIMPATH, Jung2013IRIE}. We will review these algorithms in Section~3.

As a consensus~\cite{Arora2017Debunking, Li2018Influence}, a desirable IM algorithm that is applicable in practice should attain four properties at the same time: \textbf{\em 1) high time efficiency}, \textbf{\em 2) good result quality}, \textbf{\em 3) low memory footprint}, and \textbf{\em 4) high robustness.}
However, according to the benchmarking study~\cite{Arora2017Debunking}, none of the existing IM algorithms can satisfy all these criteria at the same time.
In particular, the simulation-based algorithms are known to be computationally expensive.
They often run in $O(knm\epsilon^{-2})$ time, where $n$ and $m$ are the number of vertices and edges in the network, respectively, and $\epsilon \in (0, 1)$.
The reverse sampling algorithms have to store all samples in the main memory for seed selection, so their memory overheads are often very high.
The score estimation algorithms either take an enormous amount of time or produce results of low quality.

The robustness evaluation in the literature is inadequate because it only focuses on the performance of IM algorithms on different social networks with various structures~\cite{Arora2017Debunking, Galhotra2016Holistic, Nguyen2016Stop, Tang2015Influence, Tang2014Influence}.
In addition to this, it is also important to evaluate IM algorithms for various influence probability settings because influence probabilities are key components of influence networks, and they can directly affect the performance of IM algorithms.
As verified by our experiments, the execution time and the memory footprint of the reverse sampling algorithms are very sensitive to influence probabilities because their sample size grows exponentially as influence probabilities become larger.

In Section~3, we carry out a detailed study on how the existing IM algorithms satisfy the four properties and summarize our evaluation results in Table~\ref{Tab: IMAlgCmp}. This study motivates us to design a versatile IM algorithm that is fast, accurate, memory-efficient, and robust.

%

In this paper, we propose a new IM algorithm, called \textsf{QuickIM}.
To the best of our knowledge, \textsf{QuickIM} is the first IM algorithm that attains all the four desirable properties and is able to solve the IM problem on a network with billions of edges in several minutes.
In essence, \textsf{QuickIM} is a score estimation algorithm. It estimates the influence of every vertex by a score that is very easy to compute. In every iteration of the algorithm, it selects the vertex with the highest score as a new seed, removes it from the network and re-computes the scores of all remaining vertices. It terminates when $k$ seeds have been found.

Unlike the traditional score estimation algorithms, \textsf{QuickIM} is extremely fast and accurate. This results from two key techniques. First, \textsf{QuickIM} estimates the influence of a vertex by aggregating the probabilities of walks starting from the vertex. This score function provides a good estimate of the influence of a vertex. Second, the score of a vertex can be updated incrementally by accessing its $L$-hop neighborhood rather than visiting the entire network, where $L = 3$ is sufficient to yield highly accurate results.
The time complexity of \textsf{QuickIM} is $O(Lm + kLn + kL^{2}n'+ k\Delta^{L})$, where $n$ is the number of vertices, $m$ is the number of edges, $\Delta$ is the average in-degree of vertices, and $n'$ is far less than $n$. Since $L$ and $\Delta$ are often very small for real-world social networks, for fixed $k$,  \textsf{QuickIM} attains the $\Omega(m+n)$ lower bound on the time complexity of an IM algorithm~\cite{Borgs2012Maximizing}.
The space complexity of \textsf{QuickIM} is $O(Ln + k\Delta^{L})$, which is close to $\Omega(n)$ on real social networks.
Moreover, \textsf{QuickIM} only carries out two kinds of primitive operations, namely graph traversal and arithmetic computations. Hence, its time and space overheads are totally independent of the influence probabilities of edges in a network.

We compared \textsf{QuickIM} with the state-of-the-art IM algorithms on a variety of real social networks. 
The experimental results verify that \textsf{QuickIM} attains all the four desirable properties:

\begin{itemize}
\item {\textbf{\textsf{Fast:}}}
\textsf{QuickIM} is 1--3 orders of magnitude faster than the state-of-the-art IM algorithms.
On the largest network \textsl{Friendster} that contains more than 3.6 billion edges, \textsf{QuickIM} is able to find 100 most influential users in less than 4 minutes, while all the existing algorithms cannot terminate in an hour.

\item {\textbf{\textsf{Accurate:}}}
\textsf{QuickIM} is able to produce as good quality results as the state-of-the-art algorithms in terms of the expected fraction of influenced vertices.
The differences are all less than 0.5\% and less than 0.1\% in most cases.

\item {\textbf{\textsf{Memory-Efficient:}}}
Except \textsf{EasyIM}~\cite{Galhotra2016Holistic}, \textsf{QuickIM} requires 1--2 orders of magnitude less memory than the state-of-the-art IM algorithms. For the largest network \textsl{Friendster} in our experiments, \textsf{QuickIM} only requires 3GB of main memory in addition to the main memory for storing the network.

\item{\textbf{\textsf{Robust:}}}
The time and memory performance of \textsf{QuickIM} is very stable no matter how influence probabilities are varied.
However, the simulation-based and the reverse sampling IM algorithms are very sensitive to influence probabilities.
\end{itemize}

The rest of the paper is organized as follows:
Section~2 introduces the basic concepts and formalizes the IM problem.
Section~3 reviews the existing IM algorithms.
Sections~4 and~5 present the key techniques of \textsf{QuickIM}, namely score estimation and score updating.
Section~6 describes the procedure of \textsf{QuickIM}.
Section~7 reports the experimental evaluation. Section 8 concludes the paper.


\section{Preliminaries}
\label{Sec: Prelim}

In this section, we introduce some basic notations and concepts and formalize the Influence Maximization (\textsc{IM}) problem.

\myskip
\noindent{\textbf{\underline{Influence Networks.}}}
An \emph{influence network} is modeled as a directed and weighted graph $G = (V, E, P)$, where $V$ is a set of vertices, $E \subseteq V \times V$ is a set of edges, and
$P: E \to (0, 1]$ is a function assigning each edge $(u, v)$ with an \emph{influence probability} $P(u, v)$, i.e., the likelihood that vertex $u$ successfully influences vertex $v$.
Let $V(G)$, $E(G)$ and $P_{G}$ denote the vertex set, edge set and influence probability function of graph $G$, respectively.
Let $n = |V(G)|$ and $m = |E(G)|$.
For each edge $(u, v) \in E(G)$, $u$ is an in-neighbor of $v$, and $v$ is an out-neighbor of $u$.
Let $N^{I}_{G}(v)$ and $N^{O}_{G}(v)$ be the set of in-neighbors and out-neighbors of vertex $v$ in graph $G$, respectively. Let $d^{I}_{G}(v) = |N^{I}_{G}(v)|$ and $d^{O}_{G}(v) = |N^{O}_{G}(v)|$ be the in-degree and out-degree of $v$ in $G$, respectively.

\myskip
\noindent{\textbf{\underline{Influence Diffusion Process.}}}
Let $S \subseteq V(G)$ be a set of seed vertices. Starting from $S$, an influence diffusion process under the \emph{Independent Cascade (IC) model}~\cite{Kempe2003Maximizing} can be described as the following discrete-time stochastic process:

\begin{enumerate}\setlength{\itemsep}{0pt}\setlength{\parsep}{0pt}
\item At step $0$, we set all the seed vertices in $S$ to be \emph{active} and set all the other vertices to be \emph{inactive}. Once a vertex is activated, it remains to be active in subsequent steps.

\item At step $i$ $(i \geq 1)$, every vertex $u$ whose state was changed from inactive to active at step $i - 1$ has only one chance to activate each of its inactive out-neighbors $v$ with probability $P_{G}(u, v)$. If $u$ fails to activate $v$, $u$ can never activate $v$ in subsequent steps.

\item The influence diffusion process repeats until no more vertices can be activated.
\end{enumerate}

Let $t(S)$ be the number of active vertices when the influence diffusion process terminates. Of course, $t(S)$ is a random variable.
Therefore, we use $\mathbb{E}[t(S)]$, the expected value of $t(S)$, to evaluate the \emph{influence} of $S$, which is denoted by $I_{G}(S)$ for simplicity.

\myskip
\noindent{\textbf{\underline{Reformulation under the Possible World Model.}}}
To better understand $I_{G}(S)$, we reformulate $I_{G}(S)$ based on the well-known ``\emph{Possible World Model}'' of uncertain data~\cite{Li2012Mining, Zhu2017SimRank}. The influence network $G$ is regarded as an \emph{uncertain graph}~\cite{Li2012Mining, Khan2016Towards, Zhu2015Top, Zhu2017SimRank} at this time, where influence probability $P_G(u, v)$ is regarded as the probability that edge $(u, v)$ exists in practice.
A \emph{possible world} of $G$ is a graph obtained by instantiating each edge $e \in E(G)$ independently with probability $P_G(e)$.
Therefore, a specific possible world $g$ can be obtained from $G$ with probability
\begin{equation*}
\et
\Pr(g) = \prod_{e \in E(g)} P_{G}(e) \prod_{e \in E(G) - E(g)} 1 - P_{G}(e).
\end{equation*}
Let $\G$ be the set of all possible worlds of $G$. We can easily verify that $\G$ forms a probability space because $\sum_{g \in \G} \Pr(g) = 1$.

There is a bijective mapping from $\G$ to all instances of the influence diffusion process:
Let $g$ be a possible world of $G$.
We have a corresponding instance of the diffusion process in which the attempt that vertex $u$ try to activate vertex $v$ is successful if and only if $(u, v) \in E(g)$. Obviously, vertex $v$ is activated at step $t$ if and only if there exists a shortest path of length $t$ from a seed $s \in S$ to $v$ on $g$.
Therefore, a vertex $v$ can be activated by the seeds $S$ in this instance of diffusion process if and only if $v$ is reachable from $S$ on $g$.
Let $I_{G}(S, v)$ denote the probability that vertex $v$ is active at the end of an influence diffusion process that starts from seeds $S$. According to the bijection described above, we have
\begin{equation}
\et
\label{Eqn: IGSv}
I_G(S, v) = \sum_{g \in \G | v \text{ is reachable from } S \text{ on } g} \Pr(g).
\end{equation}
Let $R_g(S)$ denote the set of vertices reachable from $S$ on possible world $g$. We can rewrite the influence $I_{G}(S)$ of $S$ as follows:
\begin{equation}
\et
\label{Eqn: IsERgS}
I_{G}(S) = \!\!\!\!\! \sum_{v \in V(G)} \!\!\! I_{G}(S, v) = \sum_{g \in \G} \Pr(g) |R_g(S)| = \mathbb{E}[|R_g(S)|].
\end{equation}

\myskip
\noindent{\textbf{\underline{The Influence Maximization (IM) Problem.}}}
Given an influence network $G$ and the budget number of seeds $k$, the IM problem under the IC model asks for the set $S^*$ of $k$ seed vertices such that $I_{G}(S^*)$ is maximized.
It has been proven that the IM problem under the IC model is NP-hard~\cite{Kempe2003Maximizing}.


\section{Existing IM Algorithms}
\label{Sec: RevAlg}

In this section, we revisit some representative IM algorithms in the literature and show their advantages and limitations.
This study not only gives us a deep insight into the existing work on IM but also guides us to design new versatile IM algorithms.

\myskip
\noindent{\textbf{\underline{Categories of IM Algorithms.}}}
Despite of the hardness of the IM problem, a large number of algorithms with or without performance guarantees have been proposed to find a suboptimal set of seeds.
Almost all these algorithms follow a greedy framework that was first proposed by Kempe et al.~\cite{Kempe2003Maximizing}. This framework exploits the following useful property: $I_{G}(S)$ is a non-decreasing \emph{sub-modular} function of $S$~\cite{Calinescu2007Maximizing}, that is, for all $S \subseteq T$ and all $v \notin T$, $I_{G}(S \cup \{v\}) - I_{G}(S) \geq I_{G}(T \cup \{v\}) - I_{G}(T)$.
According to this property, the framework adopts the following greedy strategy to find a suboptimal solution:
Generally, the algorithm starts from an empty seed set $S$ and iteratively adds to $S$ the vertex $v$ that maximizes the marginal gain $I_{G}(S \cup \{v\}) - I_{G}(S)$ until $|S| = k$.
The most critical difference between these IM algorithms is how to overcome the \#P-hardness of computing $I_G(S)$.
Therefore, we categorize the representative IM algorithms into the three collections:

\begin{itemize}
\item{\textit{Simulation-based Algorithms}:}~\textsf{GREEDY}~\cite{Kempe2003Maximizing}, \textsf{CELF}~\cite{Leskovec2007Cost}, \break \textsf{CELF++}~\cite{Goyal2011CELF}, \textsf{StaticGreedy}~\cite{Cheng2013StaticGreedy}, and \textsf{PrunedMC}~\cite{Ohsaka2014Fast}.

\item{\textit{Reverse Sampling Algorithms}:} \textsf{RIS}~\cite{Borgs2012Maximizing}, \textsf{TIM/TIM+}~\cite{Tang2014Influence}, \textsf{IMM}~\cite{Tang2015Influence}, \textsf{SSA/D-SSA}~\cite{Nguyen2016Stop}, \textsf{SKIS}~\cite{Nguyen2017Importance}, and \textsf{Coarsen}~\cite{Ohsaka2017Coarsening}.

\item{\textit{Score Estimation Algorithms}:} \textsf{IRIE}~\cite{Jung2013IRIE} and \textsf{EasyIM}~\cite{Galhotra2016Holistic}.
\end{itemize}

The simulation-based and the reverse sampling algorithms use sampling methods to approximate the influence $I_{G}(S)$ within provable errors.
The score estimation algorithms estimate $I_{G}(S)$ using some heuristic score functions that are easy to compute and can distinguish vertices of high influence from those of low influence. Other algorithms~\cite{Chen2009Efficient, Chen2010Scalable, Chen2011Scalable, Cheng2014IMRank, Cohen2014Sketch, Goyal2012SIMPATH, Li2018Influence, Liu2014Influence} are not included in our study either due to their poor performance or to their inapplicability to the IC model.

\myskip
\noindent{\textbf{\underline{Evaluation Criteria.}}}
As pointed in~\cite{Arora2017Debunking, Li2018Influence}, a desirable IM algorithm should attain four properties at the same time, namely \textbf{\emph{high time efficiency}}, \textbf{\emph{good result quality}}, \textbf{\emph{low memory footprint}}, and \textbf{\emph{high robustness}}. The first three criteria have already been well recognized and widely adopted to evaluate IM algorithms~\cite{Arora2017Debunking, Galhotra2016Holistic, Nguyen2016Stop, Tang2015Influence, Tang2014Influence}.
Nevertheless, the robustness evaluation is still inadequate in the literature as it only focuses on evaluating the performance of an IM algorithm on different social networks with various structures~\cite{Arora2017Debunking, Galhotra2016Holistic, Nguyen2016Stop, Tang2015Influence, Tang2014Influence}.
Apart from the robustness to structural properties, it is also important to evaluate the performance of an IM algorithm with respect to different influence probability settings because influence probabilities are key elements of an influence network. Influence probabilities can directly affect the performance of influence computations. As verified by our experiments, the simulation-based and the reverse sampling IM algorithms are very sensitive to influence probabilities because they sample an edge according to its influence probability. 
To the best of our knowledge, this paper is the first one to evaluate the robustness of IM algorithms in this sense.

In Sections~\ref{Sec: RevAlg-1}--\ref{Sec: RevAlg-3}, we present a comprehensive evaluation of each category of IM algorithms listed above according to the four criteria.
Our main findings will be summarized in Section~\ref{Sec: RevAlg-4}.

\begin{table*}[t]
	\centering
    \caption{Comparisons of IM Algorithms.}
    \resizebox{\textwidth}{!}
    {
    \begin{tabular}{c|c|c|c|c|c|c||cccc}
    	\hline
        \rowcolor{mygray}
         {\bf{Algorithm}} & {\bf{IM}} & {\bf{Expected Time}} & {\bf{Worst-Case Time}} & {\bf{Space}}  & {\bf{Approximation}} & {\bf{Explanation of}} & {\bf{Time}} &{\bf{ Result}}  & {\bf{Memory}} & \bf{} \\

         \rowcolor{mygray}
         {\bf{Category}}& {\bf{Algorithm}} & {\bf{Complexity}} & {\bf{Complexity}} & {\bf{Complexity}} & {\bf{Ratio}} & {\bf{Parameters}} & {\bf{Efficiency}} &{\bf{Quality}}  & {\bf{Footprint}} & \bf{Robustness} \\ \hline

            & \textsf{Greedy}~\cite{Kempe2003Maximizing} & \multirow{3}{*}{\tabincell{c}{$O(kn\zeta_{s}\theta_{s})$ or $O(kn\theta_{s} d_{O}^{D})$ }}  & \multirow{3}{*}{$O(knm\theta_{s})$} & \multirow{3}{*}{$O(m)$} & \multirow{3}{*}{\tabincell{c}{$1 - 1/e - \epsilon$ \\ in theory}} & \multirow{3}{*}{$\theta_{s} = O(\epsilon^{-2}k^2n\log(n^2k))$} & \multirow{3}{*}{\XSolid} & \multirow{3}{*}{\CheckmarkBold} & \multirow{3}{*}{\CheckmarkBold} & \multirow{3}{*}{\XSolid} \\ \cline{2-2}

            Simulation- & \textsf{CELF~}\cite{Leskovec2007Cost} & & & & &  & & & \\ \cline{2-2}

            based& \textsf{CELF++~}\cite{Goyal2011CELF} & & & & & &  & & & \\  \cline{2-11}

            Algorithms & \textsf{StaticGreedy}~\cite{Cheng2013StaticGreedy} & \multirow{2}{*}{\tabincell{c}{$O(kn\zeta_{s}\theta_{s})$ or $O(kn\theta_{s} d_{O}^{D})$}}& \multirow{2}{*}{$O(knm\theta_{s})$} & \multirow{2}{*}{$O(m\theta_s)$} & \multirow{2}{*}{\tabincell{c}{$1 - 1/e - \epsilon$ \\ in theory}} & \multirow{2}{*}{$\theta_{s} = O(\epsilon^{-2}n\log{n \choose k})$ } & \multirow{2}{*}{\XSolid} & \multirow{2}{*}{\CheckmarkBold} & \multirow{2}{*}{\XSolid} & \multirow{2}{*}{\XSolid} \\ \cline{2-2}

            & \textsf{PrunedMC}~\cite{Ohsaka2014Fast} & & & & &  & & & \\ \hline

            \multirow{11}{*}{ \tabincell{c}{Reverse \\ Sampling \\ Algorithms}} & \textsf{RIS}~\cite{Borgs2012Maximizing} & \tabincell{c}{$O(\epsilon^{-3} k(m+n)\log^2n)$ \\ or $O(\theta_{r} d_{I}^{D})$} & $O(\theta_{r}m)$ & $O(\theta_{r}n)$ & \tabincell{c}{$1 - 1/e - \epsilon$ \\ in theory} & \tabincell{c}{the expected number of $\theta_{r}$ is \\ $O(\epsilon^{-2} n (\log n + \log {n \choose k} ) \textsf{OPT}_k^{-1})$} & \multirow{1}{*} {\XSolid} & \multirow{1}{*}{\CheckmarkBold} & \multirow{1}{*}{\XSolid} & \multirow{1}{*}{\XSolid} \\ \cline{2-11}

            & \textsf{TIM/TIM+}~\cite{Tang2014Influence} & \multirow{4}{*}{\tabincell{c}{$O(\epsilon^{-2} k(m+n)\log n )$ \\ or $O(\theta_{r} d_{I}^{D})$}} &  \multirow{4}{*}{$O(\theta_{r}m)$} & \multirow{4}{*}{$O(\theta_{r}n)$} & \multirow{4}{*}{\tabincell{c}{$1 - 1/e - \epsilon$ \\ in theory}} &  \tabincell{c}{the expected number of $\theta_{r}$ is \\ $O(\epsilon^{-2} n \textsf{OPT}_k^{-1} ((1-1/e)\alpha + \beta)^{2} )$} & \multirow{4}{*}{\XSolid} & \multirow{4}{*}{\CheckmarkBold} & \multirow{4}{*}{\XSolid} & \multirow{4}{*}{\XSolid}  \\\cline{2-2}

            & \textsf{IMM}~\cite{Tang2015Influence} & & & & & \tabincell{c}{where $\alpha = O(\log^{1/2} n)$ and \\ $\beta = O(({\log n + \log {n \choose k}})^{1/2})$}  & & & & \\ \cline{2-11}

           & \textsf{SSA/D-SSA}~\cite{Nguyen2016Stop} & $O(\theta_{r} d_{I}^{D})$  & $O(\theta_{r}m)$ & $O(\theta_{r}n)$ & \tabincell{c}{$1 - 1/e - \epsilon$ \\ in theory} & $\theta_{r}$ is less than that of \textsf{IMM} & {\CheckmarkBold} & {\CheckmarkBold} & {\XSolid} & {\XSolid} \\ \cline{2-11}

           & \textsf{SKIS}~\cite{Nguyen2017Importance} & $O(\theta_{r} d_{I}^{D})$ & $O(\theta_{r}m)$ & $O(\theta_{r}n)$ & \tabincell{c}{$1 - 1/e - \epsilon$ \\ in theory} & $\theta_{r}$ is less than that of \textsf{SSA/D-SSA}  & {\CheckmarkBold} & {\CheckmarkBold} & {\XSolid} & {\XSolid} \\ \cline{2-11}

           & \textsf{Coarsen}~\cite{Ohsaka2017Coarsening} & \tabincell{c}{$O(r(m+n) + {\theta}_{r}' {d}_{I}'^{D'})$} & $O(r(m+n) + \theta_{r}' m)$ & $O(m' +n' + \theta_{r}'n)$ & \tabincell{c}{$\gamma(1 - 1/e - \epsilon)$ \\ in theory} & \tabincell{c}{$r \in \mathbb{N}$ is a small input parameter; \\ $\gamma$ is decided by the graph $G$ and $k$; \\ ${\theta}_{r}'$, $m'$, $n'$, ${d}_{I}'$ and $D'$ have the same \\ meaning on the coarsened graph}   & \multirow{1}{*}{\CheckmarkBold} & \multirow{1}{*}{\CheckmarkBold} & \multirow{1}{*}{\XSolid} & \multirow{1}{*}{\XSolid} \\ \hline

           {\tabincell{c}{Score \\ Estimation}} & \textsf{IRIE}~\cite{Jung2013IRIE} & $O(k(kn_{i} + m))$ & $O(k(kn + m))$ & $O(n)$ & Low in practice & \tabincell{c}{$n_i$ is the expected number of $n_i(v)$; \\ $n_i(v) \leq n$ for each vertex $v$} &  \multirow{1}{*}{\XSolid} & \multirow{1}{*}{\XSolid} & \multirow{1}{*}{\CheckmarkBold} & \multirow{1}{*}{\CheckmarkBold} \\\cline{2-11}

         Algorithms & \textsf{EasyIM}~\cite{Galhotra2016Holistic} & $O(kL(m+n))$ & $O(kL(m+n))$ & $O(n)$ & High in practice & $L$ is the maximum length of paths & {\XSolid} & {\CheckmarkBold} &  {\CheckmarkBold} & {\CheckmarkBold} \\ \cline{2-11}

            \rowcolor{myemph}
             & \tabincell{c}{\textsf{QuickIM} \\ (This paper)} &  \tabincell{c}{$O(Lm + kLn + kL^{2}n'+ k\Delta^{L})$ \\ $\Omega(m + n)$ in practice} & $O(kLm + kL^{2}n')$ & \tabincell{c}{$O(Ln + k\Delta^{L})$ \\ $\Omega(n)$ in practice} & \tabincell{c}{High in practice} & \tabincell{c}{$L$ is the maximum length of walks; \\ $\Delta$ is average in-degree of all vertices; \\$n'$ is a number far less than $n$} &   \multirow{1}{*}{\CheckmarkBold} & \multirow{1}{*}{\CheckmarkBold} &  \multirow{1}{*}{\CheckmarkBold} & \multirow{1}{*}{\CheckmarkBold} \\ \hline
    \end{tabular}
    }
	\vspace{-2em}
	\label{Tab: IMAlgCmp}
\end{table*}

\subsection{Simulation-based Algorithms}
\label{Sec: RevAlg-1}

\myskip
In the simulation-based algorithms, the influence $I_G(S)$ of seeds $S$ is estimated using the following Monte-Carlo sampling method:
It samples $\theta_{s}$ possible worlds $g$ of $G$ independently at random and computes the mean of $|R_{g}(S)|$ over the sampled possible worlds, where $R_g(S)$ is the set of vertices reachable from $S$ in $g$. By Eq.~\eqref{Eqn: IsERgS}, the mean is an unbiased estimator of $I_{G}(S)$.
The simulation-based algorithms follow the greedy framework~\cite{Kempe2003Maximizing}.
In each iteration, the algorithms add to $S$ the vertex $v^*$ with the highest marginal gain $I_G(S \cup \{v^*\}) - I_G(S)$.
In particular, for each vertex $v$, the algorithms \textsf{GREEDY}~\cite{Kempe2003Maximizing}, \textsf{CELF}~\cite{Leskovec2007Cost} and \textsf{CELF++}~\cite{Goyal2011CELF} sample $\theta_{s}$ possible worlds $g$ of $G$. On each sample $g$, they conduct a breadth-first search to find $R_{g}(\{v\})$ and $R_{s}(S)$. The marginal gain $I_G(S \cup \{v\}) - I_G(S)$ is then estimated by the mean of $|R_{g}(\{v\}) - R_{g}(S)|$ over all samples $g$.
In order to improve time efficiency, the algorithms \textsf{StaticGreedy}~\cite{Cheng2013StaticGreedy} and \textsf{PrunedMC}~\cite{Ohsaka2014Fast} sample and store $\theta_{s}$ possible worlds in advance and estimate the marginal gain for each vertex $v$ on the set of pre-sampled possible worlds.

We report our evaluation results as follows.

\begin{itemize}
\item{\textit{\underline{Time Efficiency:}}}
Let $\zeta_{s}$ denote the expected time to find $R_g(\{v\})$, the vertices reachable from vertex $v$. The expected and the worst-case time complexities of all the simulation-based algorithms are $O(kn\theta_{s}\zeta_{s})$ and $O(knm\theta_{s})$, respectively. As discussed in~\cite{Tang2014Influence}, $\theta_{s}$ is often very large, so the simulation-based algorithms are costly and not scalable to large graphs as reported in~\cite{Arora2017Debunking}.

\item{\textit{\underline{Result Quality:}}}
All the simulation-based algorithms attain an approximate ratio of $1 - 1/e - \epsilon$, where $e$ is the base of natural logarithms, and $\epsilon \in (0, 1)$ is determined by $\theta_{s}$.
In practice, the simulation-based algorithms can yield results of good quality.

\item{\textit{\underline{Memory Footprint:}}}
\textsf{GREEDY}, \textsf{CELF} and \textsf{CELF++} require at most $O(m)$ space to store each sampled possible world, so their memory footprints are not high.
However, \textsf{StaticGreedy} and \textsf{PrunedMC} require $O(m\theta_{s})$ space to store all sampled possible worlds, so they have high memory overheads.

\item{\textit{\underline{Robustness:}}}
Let $d^{O}_{\G}(v) = \mathbb{E}[d^{O}_{g}(v)]$ be the expected out-degree of vertex $v$ across all possible worlds of $G$, and let $d_{O}$ be the mean of $d^{O}_{\G}(v)$ over all vertices $v$ in $G$.
We have $\zeta_{s} = O(d_{O}^{D})$, where $D$ is the diameter of $G$.
Clearly, the running time of the simulation-based algorithms is exponential in $d_{O}$, which becomes even larger when large influence probabilities are assigned to edges.
In the worst case, we have $O(d_{O}^{D}) = O(m)$.
Hence, the simulation-based algorithms are not robust to various influence probability settings.
\end{itemize}

\subsection{Reverse Sampling Algorithms}
\label{Sec: RevAlg-2}

Since the \textsf{RIS} algorithm~\cite{Borgs2012Maximizing}, the IM research community has witnessed a boom of reverse sampling algorithms. This category of algorithms use a new method to estimate the influence $I_G(S)$, which is based on a key concept called \emph{random reverse reachable set} (\emph{random RR set} for short).
Let $v$ be a vertex picked from $G$ uniformly at random and $g$ be a possible world sampled from $\G$ with probability $\Pr(g)$.
The set of vertices that can reach $v$ on $g$, denoted by $RR_{g}(v)$, forms a random RR set.
Borgs et al.~\cite{Borgs2012Maximizing} prove that the probability that the seed set $S$ overlaps with a random RR set is proportional to $I_{G}(S)$.
Therefore, a reverse sampling algorithm first samples $\theta_{r}$ random RR sets.
Then, in each iteration of the algorithm, it greedily selects the vertex that occurs in the most sampled random RR sets not overlapping with $S$ as the vertex $v^*$ with the highest marginal gain $I_{G}(S \cup \{v^*\}) - I_{G}(S)$. This is because the occurrence frequency is an unbiased estimator proportional to $I_{G}(S \cup \{v^*\}) - I_{G}(S)$~\cite{Borgs2012Maximizing}.

We report our evaluation results as follows.

\begin{itemize}
\item{\textit{\underline{Time Efficiency:}}}
The first reverse sampling algorithm \textsf{RIS}~\cite{Borgs2012Maximizing} must sample a huge number of random RR sets, so it is not efficient.
Subsequently, \textsf{TIM/TIM+}~\cite{Tang2014Influence} and \textsf{IMM}~\cite{Tang2015Influence} try to reduce the sample size $\theta_{r}$ by the bootstrap techniques.
However, their practical time efficiency is still not high~\cite{Arora2017Debunking, Nguyen2016Stop}.
Nauyen et al.~\cite{Nguyen2016Stop} claim that their \textsf{SSA} and \textsf{D-SSA} algorithms decrease $\theta_{r}$ to the lower bound\footnote{Huang et al.~\cite{Huang2017Revisiting} revisit \textsf{SSA/D-SSA}~\cite{Nguyen2016Stop} and fix some gaps.}. In order to further reduce $\theta_{r}$, \textsf{SKIS}~\cite{Nguyen2017Importance} applies sketch sampling, and \textsf{Coarsen}~\cite{Ohsaka2017Coarsening} uses graph reduction. As reported in~\cite{Huang2017Revisiting, Nguyen2017Importance, Nguyen2016Stop, Ohsaka2017Coarsening}, \textsf{D-SSA}, \textsf{SKIS} and \textsf{Coarsen} are several orders of magnitude faster than \textsf{RIS}, \textsf{TIM/TIM+}, \textsf{IMM}, and the simulation-based algorithms.

\item{\textit{\underline{Result Quality:}}}
Except \textsf{Coarsen}, all the reverse sampling algorithms attain an approximation ratio of $1 - 1/e - \epsilon$, where $\epsilon \in (0, 1)$ is a small value determined by $\theta_{r}$.
The approximation ratio of \textsf{Coarsen} is less than $1 - 1/e - \epsilon$ by a factor of $\gamma$, where $\gamma$ is determined by the input graph $G$.
As shown in~\cite{Huang2017Revisiting, Nguyen2017Importance, Nguyen2016Stop, Ohsaka2017Coarsening}, all these algorithms can produce results of good quality.

\item{\textit{\underline{Memory Footprint:}}}
All of \textsf{RIS}, \textsf{TIM/TIM+}, \textsf{IMM}, and \textsf{SSA/D-SSA} require a large amount of main memory because they must keep all the sampled random RR sets in the main memory for seed selection. Although \textsf{SKIS} and \textsf{Coarsen} mitigate this problem, their memory footprints are still high.
As reported in~\cite{Nguyen2017Importance}, \textsf{SKIS} requires nearly 100GB of main memory to process the largest dataset \textit{Friendster} used in the experiments.
\item{\textit{\underline{Robustness:}}}
Let $d^{I}_{\G}(v) = \mathbb{E}[d^{I}_{g}(v)]$ be the expected in-degree of vertex $v$ in a randomly selected possible world $g$ and $d_{I}$ be the average of $d^{I}_{\G}(v)$ over all vertices of $G$.
It takes $O(d_{I}^{D})$ time in expectation to obtain a random RR set.
Therefore, the time complexity of the reverse sampling algorithms can be reformulated as $O(\theta_{s}d_{I}^{D})$.
Obviously, the time cost of a reverse sampling algorithm grows exponentially with $d_{I}$.
In the worst case, we have $d_{I}^{D} = O(m)$. Therefore, the reverse sampling algorithms are also very sensitive to influence probabilities.
Fig.~\ref{Fig: EXP: RUN} illustrates the execution time of the three fastest reverse sampling algorithms \textsf{D-SSA}, \textsf{SKIS} and \textsf{Coarsen} to find 100 best seed vertices. The influence probabilities on all edges are set to be a constant $p_{u} \in (0, 1)$.
We can see that the execution time of all the algorithms grows drastically with $p_{u}$.
\end{itemize}

\subsection{Score Estimation Algorithms}
\label{Sec: RevAlg-3}

The score estimation algorithms use heuristic score functions to estimate the influence of each vertex.
At the very beginning, the algorithms compute the estimated influence scores of all vertices.
In each iteration of the algorithm, the vertex with the highest score is added to the seed set $S$ and is then removed from $G$.
Then, the scores of the remaining vertices are updated.
Hence, the performance of a score estimation algorithm significantly relies on the score estimation and the score updating functions. In the \textsf{IRIE} algorithm~\cite{Jung2013IRIE}, given $\theta_i \in (0, 1)$ and a vertex $v \in V(G)$, let $n_{i}(v)$ be the number of vertices $u$ such that there is a path from $u$ to $v$ that exists with probability no less than $\theta_i$. Naturally, $n_i(v)$ is large if $v$ has a high influence, so the estimated score of $v$ is proportional to $n_i(v)$.
The \textsf{EasyIM} algorithm~\cite{Galhotra2016Holistic} estimates the influence of $v$ by combining the probabilities of all simple paths that out-bound from $v$ and have length no greater than $L$. Both \textsf{IRIE} and \textsf{EasyIM} compute initial scores in a way similar to PageRank~\cite{Page1998The}.
Unfortunately, in each iteration of these algorithms, they recompute the estimated scores of all remaining vertices from scratch.

We report our evaluation results as follows.

\begin{itemize}
\item{\textit{\underline{Time Efficiency:}}}
Let $n_{i}$ be the expected number of $n_i(v)$ across all vertices.
The expected and the worst-case time complexities of \textsf{IRIE} are $O(k(kn_i + m))$ and $O(k(kn + m))$, respectively.
The expected and the worst-case time complexities of \textsf{EasyIM} are both $O(kLm)$.
Although the time complexities seem to be low, the algorithms actually run much slower than the reverse sampling algorithms like \textsf{IMM} and \textsf{D-SSA}~\cite{Arora2017Debunking, Wang2017Bring}.
This is because their score updating process is expensive, which must scan the input graph $G$ multiple times.

\item{\textit{\underline{Result Quality:}}}
Although there are no theoretical guarantee on the approximation ratios of the score estimation algorithms, the quality of results is generally high in practice.
As evaluated in~\cite{Arora2017Debunking, Galhotra2016Holistic}, the result quality of \textsf{EasyIM} is  almost the highest among all the tested IM algorithms.
However, the result quality of \textsf{IRIE} is not good~\cite{Arora2017Debunking}.
This is because the score estimation function used by \textsf{IRIE} only considers paths with the maximum existing probabilities but neglects many paths with useful information.

\item{\textit{\underline{Memory Footprint:}}}
\textsf{IRIE} and \textsf{EasyIM} only need to store the scores of vertices, so their space complexities are all $O(n)$, the lowest among all the IM algorithms.

\item{\textit{\underline{Robustness:}}}
Both \textsf{IRIE} and \textsf{EasyIM} are insensitive to influence probabilities.
This is because the influence score estimation process only performs arithmetic computations, and the computation time is independent of influence probability values.
In \textsf{IRIE}, although $n_i$ may become larger for large influence probabilities, it is much smaller than the number of edges,
which dominates the execution time of \textsf{IRIE}. Hence, \textsf{IRIE} also attains high robustness.
\end{itemize}

\subsection{Summary and Motivation}
\label{Sec: RevAlg-4}

We give a summary of the evaluation results in the last four columns of Table~\ref{Tab: IMAlgCmp}.
We have the following observations.

\begin{enumerate}
\item None of the existing IM algorithms can attain the four desirable properties at the same time.

\item Both the simulation-based algorithms and the reverse sampling algorithms are sensitive to influence probabilities and usually have high memory overheads. These algorithms are able to find high quality seeds.
The running time of the reverse sampling algorithms are determined by the sample size.

\item The score estimation algorithms are insensitive to influence probabilities. They have low memory footprint and can produce high quality results when good score estimation functions are used.
Their time efficiency is low because influence scores are repeatedly computed from scratch.

\end{enumerate}

\myskip
\noindent{\textbf{\underline{Motivation.}}} These findings motivate us to design a versatile IM algorithm, which should be \textbf{fast}, \textbf{accurate}, \textbf{memory-efficient}, and \textbf{robust} at the same time.
Since the sampling-based approach is inevitably sensitive to influence probabilities, the most promising way is to design a new score estimation algorithm. We have two design goals:
\begin{enumerate}
\item Design an accurate score estimation function with practically high quality.
\item Design an efficient score updating function, which needs not to access the whole graph from scratch in each iteration.
\end{enumerate}

In this paper, we propose a new IM algorithm, called \textsf{QuickIM}.
The properties of this algorithm are highlighted in the last row of Table~\ref{Tab: IMAlgCmp}.
To the best of our knowledge, \textsf{QuickIM} is the first IM algorithm that attains the four properties at the same time.
In Sections~\ref{Sec: InfEst} and~\ref{Sec: ScrUpt}, we describe how \textsf{QuickIM} estimates and updates influence scores, respectively.

\section{Influence Score Estimation}
\label{Sec: InfEst}

In this section, we introduce our influence score function to estimate the influence of a vertex.
As shown in~\cite{Arora2017Debunking, Galhotra2016Holistic}, the aggregated probabilities of all simple paths starting from a vertex $v$ (i.e., path probabilities) is a good estimate of $v$'s influence.
However, it is difficult to update this score efficiently.
In this paper, we use the aggregated probabilities of all \emph{walks} starting from $v$ (i.e., walk probabilities) to estimate $v$'s influence.
Our new score function is not only easy to compute and update but also can produce high quality estimation.
Section~\ref{Sec: InfEst-1} introduces the walk probability concept, and Section~\ref{Sec: InfEst-2} proposes the influence score computation method.

\subsection{Walk Probability}
\label{Sec: InfEst-1}

In this subsection, we formulate the concept of walk probability, which is the basis of our new influence score function.

Let $W = (v_0, v_1, \dots, v_t)$ be a sequence of vertices of the graph $G$. $W$ is a \emph{walk} if $(v_i, v_{i+1})$ is an edge of $G$ for all $0 \leq i \leq t-1$. The length of $W$ is $t$.
For brevity, we use $v \in W$ to denote that $W$ goes through vertex $v$, and use $(u, v) \in W$ to denote that $W$ goes through edge $(u, v)$.
Notably, a walk may go through an edge multiple times.
Hence, we use $\alpha_{W}(u, v)$ to represent how many times walk $W$ goes through edge $(u, v)$.

\myskip
\noindent{\textbf{\underline{Probability Space of Walk Probability.}}}
In the probability theory, a probability function is defined in a probability space. Therefore, we first define the probability space of walk probability.
This probability space is the foundation for correctly formulating and computing walk probabilities.

Without loss of generality, we only consider walks of length at most $L$.\footnote{We set $L = +\infty$ if all walks need to be considered.}
Therefore, an edge can be traversed by a walk at most $L$ times.
To explicitly represent multiple occurrences of an edge, we construct a multi-graph $G^{L}$ based on $G$. In a multi-graph, there may exist multiple edges from a vertex to another one. In particular, we have $V(G^{L}) = V(G)$; For each edge $(u, v) \in E(G)$, we have $L$ distinct edges from $u$ to $v$ in $G^{L}$, which are denoted by ${(u, v)}^{1}, {(u, v)}^{2}, \dots, {(u, v)}^{L}$, respectively.
For any edge $e \in E(G)$, the probabilities of the edges $e^{1}, e^2, \dots, e^L$ in $G^L$ are defined as follows:

\begin{equation}
\et
\label{Eqn: PGLei}
P_{G^{L}}(e^{i}) =
\begin{cases}
P_{G}(e) & \text{ if } i = 1, \\
P_{G}(e) & \text{ if } i > 1 \text{ and } e^{i-1} \text{ exists}, \\
0 & \text{ if } i > 1 \text{ and } e^{i-1} \text{ does not exist}. \\
\end{cases}
\end{equation}
Clearly, the probability of $e^{i}$ conditionally depends on the existence of $e^{i-1}$ for all $2 \le i \le L$.

According to the possible world model~\cite{Li2012Mining, Zhu2017SimRank}, a possible world $g^{L}$ of the multi-graph $G^{L}$ is also a multi-graph.
Let $\alpha_{g^{L}}(u, v)$ be the number of edges from $u$ to $v$ in $g^{L}$.
Assume that the existence of the edges of $G^L$ are mutually independent.
The existence probability of $g^{L}$, denoted by $\Pr(g^{L})$, is given by the following lemma. Due to space limits, we put the proof of all the lemmas in Appendix~A of the full version of this paper~\cite{Zhu2018FullVersion}.

\begin{lemma}
\label{Lem: gLPr}
For any possible world $g^{L}$ of $G^{L}$, we have
\begin{equation}
\et
\Pr(g^{L})  = \!\!\!\!\!\!\!\!  \prod_{e \in E(G)| \alpha_{g^{L}}(e) > 0} \!\!\!\!  {P_{G}(e)}^{\alpha_{g^{L}}(e)} \!\!\!\!\!\!\!\!   \prod_{e \in E(G)| \alpha_{g^{L}}(e) < L} \!\!\!\!\!\!\!\!   1 - P_{G}(e).
\end{equation}
\end{lemma}

Let ${\G}^{L}$ be the set of all possible worlds of $G^{L}$.
${\G}^{L}$ forms the sample space of the probability space, and the function $\Pr(\cdot)$ defined in Lemma~\ref{Lem: gLPr} is a probability function over ${\G}^{L}$.

\myskip
\noindent{\textbf{\underline{Walk Probability.}}}
We are now ready to formulate the concept of \emph{walk probability}. Let $g^{L}$ be a possible world of ${G}^{L}$.
A walk $W$ is said to be \emph{embedded} in $g^{L}$ if $\alpha_{g^{L}}(e) \geq \alpha_{W}(e)$ for all edges $e \in W$.
In other words, when we traverse on $g^{L}$ following the edge sequence of $W$, we can pick a distinct edge to go at each step of $W$.
The probability that $W$ is embedded in a possible world (the probability of $W$ for short), denoted by $\Pr(W)$, is thus given by
\begin{equation*}
\et
\Pr(W) = \sum_{g^L \in \G^L | W \text{ is embedded in } g^L} \Pr(g^L).
\end{equation*}
The following lemma gives us an easy way to compute $\Pr(W)$ in polynomial time.

\begin{lemma}
\label{Lem: WalkPr}
For a walk $W = (v_0, v_1, \dots, v_t)$ in the graph $G$,
\begin{equation}
\et
\Pr(W) = \prod_{i = 0}^{t-1} P_{G}(v_{i}, v_{i+1}) = \prod_{(u, v) \in W} {P_{G}(u, v)}^{\alpha_{W}(u, v)}.
\end{equation}
\end{lemma}

Next, we extend the above definition and formulate the probability of multiple walks.
Given a set of walks $W_1, W_2, \dots, W_t$, let $\Pr(\bigwedge_{j=1}^{t}W_j)$ represent the probability that all these walks are embedded in a randomly selected possible world of $G^{L}$, that is,
\begin{equation*}
\et
\Pr(\bigwedge_{j=1}^{t}W_j) = \sum_{g^L \in \G^L | W_1, W_2, \dots, W_t \text{ are embedded in } g^L} \Pr(g^L).
\end{equation*}
The following lemma gives us a polynomial-time method to compute $\Pr(\bigwedge_{j=1}^{t}W_j)$.

\begin{lemma}
\label{Lem: MWalkPr}
For several walks $W_1, W_2, \dots, W_t$ in the graph $G$,
\begin{equation}
\et
\Pr(\bigwedge_{j=1}^{t}W_j) = \!\!\!\!\!\!\!\!\!\!\!\! \prod_{(u, v) \text{ is an edge in any of } W_1, W_2, \dots, W_t} \!\!\!\!\!\!\!\!\!\!\!\!\!\!\!\!\!\!\!\!\!\!\!\! {P_{G}(u, v)}^{\max_{1 \leq i \leq t}\alpha_{W_i}(u, v)}.
\end{equation}
\end{lemma}

According to Lemma~\ref{Lem: MWalkPr}, the event that a walk $W$ is embedded in a possible world $g^{L}$ is not independent of the event that a walk $W'$ is embedded in $g^{L}$ if $W$ and $W'$ go through some common edges.
%

\subsection{Influence Score Function}
\label{Sec: InfEst-2}

In this subsection, we formulate our new influence score function based on the walk probability concept.
For simplicity of notation, when there is only one vertex in the seed set, we use $I_{G}(v)$ to simply denote $I_{G}(\{v\})$, the influence of the seed set $\{v\}$, and we use $I_{G}(u, v)$ to simply denote $I_{G}(\{u\}, v)$, the influence of the seed set $\{u\}$ on vertex $v$.

Given two vertices $u$ and $v$, suppose there are $h_{uv}$ walks $W_1, W_2,\break \dots, W_{h_{uv}}$ that start from $u$, end at $v$, and are of length at most $L$. Let $\Pr(\bigvee_{j=1}^{t}W_j)$ represent the probability that at least one of the walks $W_1, W_2, \dots, W_{h_{uv}}$ is embedded in a randomly selected possible world of $G^L$.
Let $[h_{uv}]$ be a concise representation of the set $\{1, 2, \dots, h_{uv}\}$.
The following lemma states the relationship between $I_{G}(u, v)$ and the probabilities of $W_1, W_2, \dots, W_{h_{uv}}$.

\begin{lemma}
\label{Lem: IGPrW}
\begin{equation}\label{eqn:IGPrW}
\et
\begin{split}
~ & I_{G}(u, v) = \Pr(\bigvee_{i = 1}^{h_{uv}} W_i) = \sum_{i=1}^{h_{uv}} \Pr(W_{i}) - \!\!\!\! \!\!\!\!  \sum_{1 \leq i < j \leq h_{uv}} \!\!\!\! \!\!\!\! \Pr(\bigwedge_{i =1}^{h_{uv}} W_i) + \dots \\
& + (-1)^{t-1} \!\!\!\!\!\!\!\!  \sum_{C \subseteq [h_{uv}], |C| = t} \!\!\!\! \!\!\!\!  \Pr(\bigwedge_{i \in C} W_i) + \dots +
(-1)^{h_{uv}-1} \Pr(\bigwedge_{i = 1}^{h_{uv}} W_i).
\end{split}
\end{equation}
\end{lemma}

Despite the complexity of Eq.~\eqref{eqn:IGPrW}, the first term $\sum_{i=1}^{h_{uv}} \Pr(W_{i})$ can be used as a good estimate of $I_{G}(u, v)$ in practice. For simplicity of notation, let $W_{G}(u, v) = \sum_{i=1}^{h_{uv}} \Pr(W_{i})$. By Eq.~\eqref{Eqn: IsERgS}, we have $I_G(u) = \sum_{v \in V(G)} I_G(u, v)$. Thus, we can use $\widehat{I}_{G}(u) = \sum_{v \in V(G)} W_{G}(u, v)$ as an estimated score of $u$'s influence $I_G(u)$.

\myskip
\noindent{\textbf{\underline{Correlation Analysis.}}}
In score estimation algorithms, the score of a vertex is expected to be highly correlated with its true influence,
that is, a vertex with higher influence should has a higher score. With this, the algorithm is able to select good seeds with high influence and produce high quality results.
In the following, we show that our score function $\widehat{I}_{G}(\cdot)$ meets this requirement.

First, we reformulate $W_{G}(u, v)$ and $I_{G}(u, v)$. Let $X^{\langle t \rangle}_{G}(u, v)$ be the sum of probabilities $\Pr(g^{L})$ of possible worlds $g^{L}$, in which there are embedded exactly $t$ walks from $u$ to $v$ of length at most $L$.
According to the following lemma, we can rewrite $I_{G}(u, v)$ and $W_{G}(u, v)$ as functions of $X^{\langle i \rangle}_{G}(u, v)$.

\begin{lemma}
We have
\label{Lem: XRwIW}
$I_{G}(u, v) = \sum_{i=1}^{h_{uv}} X^{\langle i \rangle}_{G}(u, v)$, and $W_{G}(u, v) \break = \sum_{i=1}^{h_{uv}} i X^{\langle i \rangle}_{G}(u, v)$.
\end{lemma}

From Lemma~\ref{Lem: XRwIW}, we have $I_{G}(u, v) \leq W_{G}(u, v) \leq h_{uv} I_{G}(u, v)$. Therefore, when $I_{G}(u, v)$ is large, $W_{G}(u, v)$ is also large. Thus,
$W_{G}(u, v)$ can distinguish vertices of high influence from those of low influence.

Furthermore, we analyze the gap between the score $\widehat{I}_{G}(u)$ of vertex $u$ and its true influence $I_G(u)$.
Let $p_{m} = \max_{e \in E(G)} P_{G}(e)$ be the maximum influence probability of all edges in $G$.
Obviously, when $i \geq 2$, there exist at least 3 distinct edges in a possible world graph that embeds $i$ walks.
Therefore, we have $X^{\langle i \rangle}_{G}(u, v) \leq {h_{uv} \choose i} p_{m}^{3}$.
Let $\epsilon_{G}(u, v) = W_{G}(u, v) - I_{G}(u, v)$.
By Lemma~\ref{Lem: XRwIW}, $\epsilon_{G}(u, v) \geq 0$ obviously holds. We also have
\begin{equation}
\et
\label{Eqn: Epsiuv}
\begin{split}
\epsilon_{G}(u, v)  \! &  = \! \sum_{i=2}^{h_{uv}} (i - 1) \!X^{\langle i \rangle}_{G}(u, v)  \! \le \!  (h_{uv} - 1) \sum_{i=2}^{h_{uv}}  \! X^{\langle i \rangle}_{G}(u, v) \\
& \leq  p_{m}^3 (h_{uv} - 1) \sum_{i=2}^{h_{uv}} {h_{uv} \choose i}  \leq p_{m}^{3} h_{uv} 2^{h_{uv}}.
\end{split}
\end{equation}
By Eq.~\eqref{Eqn: Epsiuv}, we have the following bounds on $\widehat{I}_{G}(u) - I_G(u)$.
\begin{equation}
\et
\label{Eqn: IGuBound}
0 \leq \widehat{I}_{G}(u) - I_{G}(u) \leq p_{m}^3 L \sum_{v \in V(G)} h_{uv} 2^{h_{uv}}.
\end{equation}
From Eq.~\eqref{Eqn: IGuBound}, we find that the gap between $\widehat{I}_{G}(u)$ and ${I}_{G}(u)$ is determined by $h_{uv}$.
We can easily see that this gap is highly correlated with ${I}_{G}(u)$, that is, when $I_{G}(u)$ is large, $u$ is likely to connect to many other vertices $v$ with high influence probabilities $I_{G}(u, v)$.
Since $I_{G}(u, v)$ is large, there may exist multiple paths from $u$ to $v$, so the number of walks $h_{uv}$ is also large.
Thus, the estimation $\widehat{I}_{G}(u)$ tends to enlarge ${I}_{G}(u)$ when ${I}_{G}(u)$ is large.
According to this property, we can easily identify vertices with high influence.
As verified by the experimental results in Section~\ref{Sec: PerEva}, by using this score function, our algorithm can produce as high quality results as the state-of-the-art IM algorithms.

%

\myskip
\noindent{\textbf{\underline{Revisiting Error Analysis in~\cite{Galhotra2016Holistic}.}}}
In addition to our analysis, we revise some existing results in~\cite{Galhotra2016Holistic}.
Note that if the input graph $G$ is a directed acyclic graph (DAG), every walk from a vertex $u$ to a vertex $v$ must be a simple path from $u$ to $v$. In this case, Lemma~4 and Lemma~5 in~\cite{Galhotra2016Holistic} derive two relative errors of $W_{G}(u, v)$, namely
$\epsilon_{1}^{DAG} = \sum_{w \in N_{G}^{I}(v)} (p_{(w, v)} - 1)A_{1}$ and $\epsilon_{2}^{DAG} = \sum_{p \in \mathbb{P}_{uv}} \prod_{e \in p} p_e$, where $A_1 = \sum_{p \in \mathbb{P}_{uw}} \prod_{e \in p} p_e$ and $\mathbb{P}_{uv}$ is the set of all walks (i.e., simple paths) from $u$ to $v$ on $G$.
According to Eq.~(7) in~\cite{Galhotra2016Holistic}, $\epsilon_{1}^{DAG} \ge 0$. However, since $A_1 > 0$ and $p_{(w, v)} < 1$ for some edges $(w, v) \in E(G)$,
we have $\epsilon_{1}^{DAG} < 0$, which is a contradiction. This mistake is caused by the following reason: In~\cite{Galhotra2016Holistic}, it is regarded that $\Pr(\bigwedge_{j=1}^{t}W_j) = \prod_{j=1}^{t} \Pr(W)$. However, the equation only holds when all of
$W_i, W_2, \dots, W_t$ have no common edges.

\myskip
\noindent{\textbf{\underline{Score Computation Method.}}}
The influence scores $\widehat{I}_{G}(u)$ of all vertices $u$ in the graph $G$ can be computed as follows:
Let $\A_{n \times n}$ be the adjacency matrix of $G$, where $\A[u, v] = P_{G}(u, v)$ if $(u, v) \in E(G)$ and $\A[u, v] = 0$ otherwise.
Naturally, we have $W_{G}(u, v) = \sum_{j=1}^{L} \A^{j}[u, v]$.
Let $\HH = (1, 1, \dots, 1)^T$ be a $n$-dimensional vector with all elements 1.
We have $\widehat{I}_{G}(u) = \sum_{j=1}^{L} \sum_{v \in V(G)}\A^{j}[u, v] = \sum_{i=1}^{L}(\A^{j}\HH)[u]$.
Let $\F_j = \A^{j}\HH$ and $\F = \sum_{j=1}^{L} \F_{j}$.
We have $\F[u] = \widehat{I}_{G}(u)$ for all vertices $u$ of $G$.

We present the \textsf{ScoreEst} procedure to compute $\F$.
First, we compute $\F_{1} = \A\HH$ (line~2). Then, we iteratively compute $\F_{j}$ by left multiplying $\F_{j-1}$ with $\A$ for $j = 2, 3, \dots, L$ (line~4). Finally, we output $\F_1 + \F_2 + \dots + \F_L$ as $\F$ (line~5). We store $\F_1, \F_2, \dots, \F_L,$ and $\F$ in the main memory because they will be frequently used later in the score updating process.

Every matrix multiplication in \textsf{ScoreEst} can be done in $O(m)$ time, so the time complexity of \textsf{ScoreEst} is $O(Lm)$. The space complexity of \textsf{ScoreEst} is $O(Ln)$ since it stores $\F_1, \F_2, \dots, \F_L$, and $\F$, each of which requires $O(n)$ space.

\begin{figure}[!t]
    \centering
    \scriptsize
    \resizebox{\algswidth}{!}{
    \fbox{
    \parbox{\figwidth}{
    {
    \textbf{Procedure} \textsf{ScoreEst}$(G, L)$
    \begin{algorithmic}[1]
	\STATE $\A \gets$ the adjacency matrix of the graph $G$, $\HH \gets (1, 1, \dots, 1)^T$
	\STATE $\F_{1} \gets \A\HH$
	\FOR{$j \gets2 $ to $L$}
		\STATE $\F_i \gets \A \F_{i-1}$  // $\F_i$ is stored in the main memory
	\ENDFOR
	\RETURN $\sum_{i=1}^{L} \F_i$
    \end{algorithmic}
    }
    }}}
    \vspace{-2em}
\end{figure}

\section{Incremental Score Updating}
\label{Sec: ScrUpt}

In this section, we present an efficient method to update the influence score of a vertex.
Recall that, after finding the vertex $w$ with the highest estimated score, \textsf{EasyIM}~\cite{Galhotra2016Holistic} and \textsf{IRIE}~\cite{Jung2013IRIE} remove $w$ from the graph $G$ and recomputes the
scores of all remaining vertices from scratch. Obviously, this updating strategy is time consuming.
We observe that the removal of $w$ only affects the scores of vertices in the proximity of $w$.
Therefore, we propose an \emph{incremental} method to fast update the influence score of each remaining vertex.
Section~\ref{Sec: ScrUpt-1} presents the basic updating method, and Section~\ref{Sec: ScrUpt-2} describes the more efficient lazy updating method.

\subsection{Basic Updating Method}
\label{Sec: ScrUpt-1}

Given $k$, the budge number of seeds, all the IM algorithms that follow the general greedy framework must carry out $k$ iterations.
For $1 \leq t \leq k$, let $\F^{(t)}$ and $\F^{(t)}_{i}$ represent the states of $\F$ and $\F_i$ in the $t$-th iteration, respectively.
Note that $\F^{(1)}$ and $\F^{(1)}_{i}$ are initial vectors computed and stored by the \textsf{ScoreEst} procedure.

Let $w^{(t)}$ be the vertex with the hightest estimated score selected in the $t$-th iteration.
Before we proceed to the $(t + 1)$-th iteration, we must remove $w^{(t)}$ from $G$, that is, remove all the edges incident to $w^{(t)}$.
Interestingly, it is sufficient to only remove all the out-edges of $w^{(t)}$ because keeping all the in-edges of $w^{(t)}$ in $G$ has an insignificant impact on score updating when $p_{m}$ is small, where $p_{m} = \max_{e \in E(G)} P_{G}(e)$ is the maximum influence probability of all edges in $G$. Moreover, removing the out-edges of $w^{(t)}$ makes score updating easier. This is guaranteed by the following lemma.

\begin{lemma}
\label{Lem: InfluInw}
Let $w$ be a vertex of $G$. Let $G_1$ be the graph obtained by removing all the incident edges of $w$ from $G$. Let $G_2$ be the graph obtained by removing all the out-edges of $w$ from $G$. For any vertex $u \neq w$, we have
$0 \leq |I_{G_{1}}(u) - I_{G_{2}}(u)| \leq 1$ and $|\widehat{I}_{G_1}(u)  - \widehat{I}_{G_2}(u) | \leq  p_{m}^{3} h_{uw} 2^{h_{uw}}$, where $h_{uw}$ is the number of walks from $u$ to $w$.
\end{lemma}

Let $\A^{(t)}$ be the adjacency matrix of the graph $G$ at the beginning of the $t$-th iteration and let $\M^{(t)}$ be a matrix, where $\M^{(t)}[u, v] = -P_{G}(u, v)$ if $u = w^{(t)}$, and $\M[u, v] = 0$ otherwise.
For $1 \leq i \leq L$, let
$\Delta \F_{i}^{(t)} = \F_i^{(t+1)} - \F_i^{(t)}$. Obviously, we have $\F_i^{(t+1)} = (\A^{(t)} + \M^{(t)})^i \HH$ and $\F^{(t+1)} = \F^{(t)} + \sum_{i=1}^{L} \Delta \F_{i}^{(t)}$.
Thus, to update $\F^{(t)}$ to $\F^{(t+1)}$, it is sufficient to compute all $\Delta \F_{i}^{(t)}$.

\myskip
\noindent{\textbf{\underline{Computating $\Delta \F^{(t)}_{i}$.}}}
First, we can expand the equation $\Delta \F^{(t)}_{i} = (\A^{(t)} + \M^{(t)})^i \HH - {\A^{(t)}}^i \HH$ by the binomial theorem as follows:
\begin{equation}
\label{Eqn: DeltaFi}
\Delta \F^{(t)}_{i} \!\!  = \!\!\!\!\!\!\!\!\!\!\!\!\!\! \sum_{\alpha_1, \dots, \alpha_s , \beta_1, \dots, \beta_{s-1}}
\!\!\!\!\!\!\!\!\!\!\!\!\!\! {\A^{(t)}}^{\alpha_1} {\M^{(t)}}^{\beta_1}  \dots \! {\A^{(t)}}^{\alpha_{s-1}} {\M^{(t)}}^{\beta_{s-1}}{\A^{(t)}}^{\alpha_s}\HH,
\end{equation}
\begin{equation*}
\scriptsize
\begin{split}
\text{\normalsize where} ~& 0 \leq \alpha_h \leq i ~ \text{\normalsize  for all } 1 \leq h \leq s, \\
~ & 0 \leq \beta_{h} \leq i ~ \text{\normalsize for all } 1 \leq h \leq s - 1, \\
~ &\sum_{h=1}^{s} \alpha_h + \sum_{h=1}^{s-1} \beta_h = i,~ \text{\normalsize and}~ \beta_h > 0, \text{\normalsize for some } 1 \leq h \leq s-1.
\end{split}
\end{equation*}
The terms in Eq.~\eqref{Eqn: DeltaFi} can be categorized into three disjoint groups and are dealt with in different ways:

\begin{itemize}
\item{\textit{{Group~1:}}} This group is composed by the terms with $\alpha_1 = 0$. In the results of these terms, all the elements except the $w^{(t)}$-th one are 0. However, since $w^{(t)}$ has already been selected as a seed, it is unnecessary to update its score any more. Hence, we can eliminate these terms from Eq.~\eqref{Eqn: DeltaFi} without affecting updating the scores of the vertices that have not been selected as seeds yet.

\item{\textit{{Group~2:}}} This group is composed by the terms with $\alpha_1 > 0$ and $\beta_j \geq 2$ for some $1 \leq j \leq s-1$. When $\beta_j \geq 2$, we have ${\M^{(t)}}^{\beta_j} = \0$ because there is no self-loop (i.e., edge from a vertex to itself) in $G$, where $\0 = (0, 0, \dots, 0)^T$ represents the $n$-dimensional vector whose elements are all 0. Hence, all terms in this category are equal to $\0$.

\item{\textit{{Group~3:}}} This group is composed by the terms with $\alpha_1 > 0$ and $0 \le \beta_{1}, \beta_{2}, \dots, \beta_{s - 1} \le 1$. This category can be further divided into $i-1$ disjoint groups according to $\alpha_1$. All the terms with the same $\alpha_1$ belong to the same group.
For $1 \leq j \leq i-1$, let
\begin{equation*}
\scriptsize
\label{Eqn: DeltaFi-2}
\Y_{j} \!\!=\!\! {\A^{(t)}}^j \M^{(t)} \!\!\!\!\!\!\!\!\!\!\!\!\!\!\!\!\! \sum_{\alpha_2, \dots, \alpha_s , \beta_2, \dots, \beta_{s-1}}
\!\!\!\!\!\!\!\!\!\!\!\!\!\!\!\! {\A^{(t)}}^{\alpha_2} {\M^{(t)}}^{\beta_2} \!\!\!\!\!\!\! \dots \! {\A^{(t)}}^{\alpha_{s-1}} {\M^{(t)}}^{\beta_{s-1}}{\A^{(t)}}^{\alpha_s} \HH,
\end{equation*}
\begin{equation*}
\scriptsize
\begin{split}
\text{\normalsize where} ~& 0 \leq \alpha_h \leq i ~ \text{\normalsize  for all } 2 \leq h \leq s, \\
~ & 0 \leq \beta_{h} \leq 1 ~ \text{\normalsize for all } 2 \leq h \leq s - 1, \text{\normalsize and}\\
~ &\sum_{h=2}^{s} \alpha_h + \sum_{h=2}^{s-1} \beta_h = i - j - 1.
\end{split}
\end{equation*}

By the binomial theorem, the terms in $\Y_{j}$ are handled according to the following two cases:
\begin{itemize}
\item{\textit{{Case~1:}}} The terms in $\Y_{j}$ with $\beta_{h} = 1$ for some $2 \leq h \leq s$ must be contained in the expansion of the equation
${\A^{(t)}}^j \M^{(t)} \Delta \F^{(t)}_{i-j-1} = {\A^{(t)}}^j \M^{(t)} ((\A^{(t)} + \M^{(t)})^{i-j-1} \HH - {\A^{(t)}}^{i-j-1} \HH) $.
\item{\textit{{Case~2:}}} The only term in $\Y_{j}$ with $\beta_{h} = 0$ for all $2 \leq h \leq s$ must be equal to ${\A^{(t)}}^j \M^{(t)}\F^{(t)}_{i-j-1} = {\A^{(t)}}^j \M^{(t)} {\A^{(t)}}^{i-j-1} \HH$.
\end{itemize}
Thus, we have $\Y_{j} = {\A^{(t)}}^j \M^{(t)} (\Delta \F^{(t)}_{i-j-1} + \F^{(t)}_{i-j-1}) $.
This implies that all terms in $\Y_{j}$ can be computed together.
Let $\Delta \F^{(t)}_0 = \Delta \F^{(t)}_1 = \0$ and $\F^{(t)}_0 = \HH$.
For $i \geq 2$, we can rewrite $\Delta \F^{(t)}_i$ as
\begin{equation}
\et
\label{Eqn: DeltaFiRec}
\Delta \F^{(t)}_i = \sum_{j=1}^{i-1} {\A^{(t)}}^{j} \M^{(t)} (\Delta \F^{(t)}_{i-j-1} + \F^{(t)}_{i-j-1}).
\end{equation}
\end{itemize}
Eq.~\eqref{Eqn: DeltaFiRec} implies that we can compute $\Delta \F^{(t)}_i$ in a recursive manner. In particular, if we already have $\Delta \F^{(t)}_0, \Delta \F^{(t)}_1, \dots, \Delta \F^{(t)}_{i-1}$,
we can easily compute $\M^{(t)} (\Delta \F^{(t)}_{i-j-1} + \F^{(t)}_{i-j-1})$ for each $j$. Note that only the $w^{(t)}$-th element of $\M (\Delta \F_{i-j-1} + \F_{i-j-1})$ is non-zero.
Let $c^{(t)}_{x}$ denote $\M^{(t)} (\Delta \F^{(t)}_{x} + \F^{(t)}_{x})[w^{(t)}]$ for $0 \leq x \leq i$. We have
\begin{equation}
\label{Eqn: cxt}
c^{(t)}_x = - \sum_{v \in N_{G}^{O}(w^{(t)})} P_{G}(w^{(t)}, v)(\Delta \F^{(t)}_{x}[v] + \F^{(t)}_{x}[v]).
\end{equation}
Given a matrix $X$, let $\mathbf{X}[*, u]$ denote the $u$-th column of matrix $\mathbf{X}$.
If the columns ${\A^{{(t)}}}^j[*, w^{(t)}]$ are known for all $1 \leq j \leq L$, we immediately have
\begin{equation}
\et
\label{Eqn: DFucxt}
\Delta \F_i[u] = \sum_{j=1}^{i-1} c^{(t)}_{i-j+1} {\A^{(t)}}^j[u, w^{(t)}]
\end{equation}
for all vertices $u \neq w^{(t)}$.

\myskip
\noindent{\textbf{\underline{Basic Score Updating.}}}
Based on the computation method of $\Delta \F^{(t)}_{i}$, we can easily obtain a basic method to update $\F^{(t)}$ to $\F^{(t+1)}$.
Due to space limits, here we provide a sketch of the basic updating procedure \textsf{ScoreUpd-Basic}. The details can be found in Appendix~B of the full version of this paper~\cite{Zhu2018FullVersion}.
First, we scan the vector $\F^{(t)}$ to obtain the vertex $w^{(t)}$.
The column ${\A^{^{(t)}}}^j[*, w^{(t)}]$ for each $1 \leq j \leq L$ is obtained by invoking the \textsf{WalkPro} procedure.
\textsf{WalkPro} applies a traversal by following the in-edges of $w^{({t})}$ to compute ${\A^{^{(t)}}}^j[u, w^{(t)}]$ for each vertex $u \in V(G)$ and each $1 \leq j \leq L$.
The details of \textsf{WalkPro} can also be found in Appendix~B in~\cite{Zhu2018FullVersion}.
After that, we can compute $c^{(t)}_{0}$ and $c^{(t)}_{1}$ by Eq.~\eqref{Eqn: cxt} (line~7).
Then, we iterate to compute $\Delta \F^{(t)}_{i}$ by Eq.~\eqref{Eqn: DFucxt} for each $2 \leq i \leq L$.
For each $i$, after obtaining $\Delta \F^{(t)}_{i}$ for all vertices, we compute $c^{(t)}_{i}$ by Eq.~\eqref{Eqn: cxt} and store it for following computations.
Finally, $\F^{(t+1)}[u]$ can be simply obtained by $\F^{(t)}[u] + \sum_{i=1}^{L} \Delta \F_{i}^{(t)}[u]$.

\myskip
\noindent{\textbf{\underline{Complexity Analysis.}}}
Let $\Delta$ be the average degree of vertices in $G$.
The expected and the worst-case time complexities of \textsf{WalkPro} are $O(\Delta^{L})$ and $O(Lm)$, respectively.
The expected and the worst-case time complexities of \textsf{ScoreUpd-Basic} are $O(L^{2}n + \Delta^{L})$ and $O(L^{2}n + Lm)$, respectively.
The expected and the worst-case space complexities of \textsf{ScoreUpd-Basic} is $O(\Delta^{L})$ and $O(Ln)$, respectively.
All the detailed analysis can be found in Appendix~B~\cite{Zhu2018FullVersion}.

When we apply \textsf{ScoreUpd-Basic} in all the $k$ iterations, the total expected and the total wort-case time cost are indeed $O(k(L^{2}n + \Delta^{L}))$ and $O(k(L^{2}n + Lm))$, respectively.
The expected space complexity is still $O(\Delta^{L})$ since we do not have to leave any data in the main memory after the procedure.

\subsection{Lazy Updating Method}
\label{Sec: ScrUpt-2}

In the \textsf{ScoreInc-Basic} procedure, we have the following observations when updating $\F^{(t)}$ to $\F^{(t+1)}$:

\begin{itemize}
\item{\textit{\underline{Observation~1:}}} If a vertex $v$ is impossible to have the maximum score among all the vertices remaining in $G$ in the $(t+1)$-th iteration, it is unnecessary to update $\F^{(t)}[v]$ to $\F^{(t+1)}[v]$.

\item{\textit{\underline{Observation~2:}}} The score of a vertex $v$ is monotonically non-increasing during the execution of \textsf{ScoreInc-Basic}, that is, we have $\F^{(i)}[u] \geq \F^{(j)}[u]$ if $i < j$.
\end{itemize}

These observations motivate us to delay updating the score of a vertex when it is necessary and skip updating the scores of many vertices during the score updating process.
To this end, we propose the following lazy score updating method.

Let $g^{(t)}_x = \sum_{h=0}^{x} c^{(t)}_x$. According to Eq.~\eqref{Eqn: DeltaFiRec}, we can reformulate $\F^{(t+1)}[u]$ as follows:
\begin{equation}
\et
\label{Eqn: Ft1Ft}
\F^{(t+1)}[u] = \F^{(t)}[u] + \sum_{j = 1}^{L} (g^{(t)}_{L - j - 1} \cdot {\A^{(t)}}^{j}[u, w^{(t)}] ).
\end{equation}
Clearly, we can fast update $\F^{(t)}[u]$ to $\F^{(t+1)}[u]$ if we already have $c^{(t)}_x$ for all $0 \leq x \leq L$ and the columns ${\A^{{(t)}}}^j[*, w^{(t)}]$ for all $1 \leq j \leq L$. To compute all $c^{(t)}_x$, Eq.~\eqref{Eqn: cxt} implies that we only have to compute $\F^{(t)}_{j}[v]$ and $\Delta \F^{(t)}_{j}[v]$ for $1 \leq j \leq L$ and for all vertices $v$ such that $(w^{(t)}, v) \in E(G)$.
Here, we show that $\F^{(t)}_{j}[v]$ and $\Delta \F^{(t)}_{j}[v]$ can be computed in a pay-as-you-go manner, which helps reduce the cost of computing $c^{(t)}_x$ for all $0 \leq x \leq L$.

 \begin{figure}[!t]
    \centering
    \scriptsize
    \resizebox{\algswidth}{!}{
    \fbox{
    \parbox{\figwidth}{
    \textbf{Procedure} \textsf{LazyF}$(t, j ,v)$
    \begin{algorithmic}[1]
        \IF{$j = 0$}
            \RETURN $0$
        \ENDIF
        \IF{$\F^{(t)}_{j}[v]$ is not stored}
            \STATE $\F^{(t)}_{j}[v] \gets \textsf{LazyF}(t - 1, j ,v)$ + \textsf{LazyDF}$(t - 1, j , v)$
            \STATE store $\F^{(t)}_{j}[v]$ in the main memory
            \STATE drop all elements $\F^{(t')}_{j}[v]$ where $t' < t$ from the main memory
        \ENDIF
        \RETURN $\F^{(t)}_{j}[v]$
    \end{algorithmic}
    }}}
    \label{Fig: IncInfOpt}
    \vspace{-1.5em}
\end{figure}
 \begin{figure}[!t]
    \centering
    \scriptsize
    \resizebox{\algswidth}{!}{
    \fbox{
    \parbox{\figwidth}{
    \textbf{Procedure} \textsf{LazyDF}$(t, j, v)$
    \begin{algorithmic}[1]
        \IF{$j = 0$ or $j = 1$}
            \RETURN 0
        \ENDIF
        \IF{$\Delta \F^{(t)}_{j}[v]$ is not stored}
            \STATE compute $\Delta \F^{(t)}_{j}[v]$ by  Eq.~\eqref{Eqn: DeltaFtC}
            \STATE store $\Delta \F^{(t)}_{j}[v]$ in the main memory
            \STATE drop all elements $\Delta \F^{(t')}_{j}[v]$ where $t' < t$ from the main memory
        \ENDIF
        \RETURN $\Delta \F^{(t)}_{j}[v]$
    \end{algorithmic}
    }}}
    \label{Fig: IncInfOpt}
        \vspace{-3em}
\end{figure}

\myskip
\noindent{\textbf{\underline{Lazy Computing of $\F^{(t)}_{j}[u]$ and $\Delta \F^{(t)}_{j}[u]$.}}}
We present the \textsf{LazyF} and \textsf{LazyDF} procedures to compute $\F^{(t)}_{j}[u]$ and $\Delta \F^{(t)}_{j}[u]$, respectively, when they are needed.
The \textsf{LazyF} procedure works as follows:
If $j = 0$, it returns $1$ according to the definition of $\F^{(t)}_{0}$ (line~2).
If $j > 1$, it checks if the value of $\F^{(t)}_{j}[v]$ has been recorded in the main memory. If so, the recorded value of $\F^{(t)}_{j}[v]$ is returned.
If it is not recorded, we recursively call \textsf{LazyF} to compute $\F^{(t-1)}_{j}[v]$ and call \textsf{LazyDF} to compute $\Delta \F^{(t-1)}_{j}[v]$ (line~4).
Note that, the value of $\F^{(1)}_{j}[v]$ must be stored in the main memory because it has been done in the \textsf{ScoreEst} procedure. Once the value of $\F^{(t)}_{j}[v]$ is computed, it is immediately stored in the main memory for later use (line~5). At this time, we can safely drop all elements $\F^{(t')}_{j}[v]$ for $t' < t$ from the main memory since they will never be accessed in the following computations (line~6). Finally, $\F^{(t)}_{j}[v]$ is returned as the result (line~7).

The \textsf{LazyDF} procedure works as follows:
If $j =0$ or $1$, it returns $0$ according to the definition of $\Delta \F^{(t)}_{0}$ and $\Delta \F^{(t)}_{1}$ (line~2).
If $j > 1$, it checks if the value of $\Delta \F^{(t)}_{j}[v]$ is stored in the main memory. If so, the stored value of $\F^{(t)}_{j}[v]$ is returned.
If it is not stored, \textsf{LazyDF} computes it by the following equation (line~4):
\begin{equation}
\et
\label{Eqn: DeltaFtC}
\Delta \F^{(t)}_{j}[v]  = \sum_{x=0}^{j-2} c^{(t)}_{x} {\A^{(t)}}^{j-x-1}[v, w^{(t)}] ).
\end{equation}
Assume that $c^{(t)}_x$ for all $0 \leq x \leq L$ and the columns ${\A^{{(t)}}}^j[*, w^{(t)}]$ for all $1 \leq j \leq L$ are stored in the main memory.
Similar to \textsf{LazyF}, \textsf{LazyDF} also stores $\Delta \F^{(t)}_{j}[v]$ in the main memory for later use and drops all elements $\Delta \F^{(t')}_{j}[v]$ where $t' < t$ from the main memory (lines~5--6). Finally, $\Delta \F^{(t)}_{j}[v]$ is returned as the result (line~7).

In each iteration of the greedy framework, either $\F^{(t)}_{j}[v]$ or $\Delta \F^{(t)}_{j}[v]$ is computed at most once for each vertex $v$ in $G$ and $1 \le j \le l$.

\begin{figure}[!t]
    \centering
    \scriptsize
    \resizebox{\algswidth}{!}{
    \fbox{
    \parbox{\figwidth}{
    {
    \textbf{Procedure} \textsf{ScoreUpd-Lazy}$(G, t, L, S)$
    \begin{algorithmic}[1]
    \STATE $w^{(t)} = \arg\max_{v \in V(G) - S, t_{v} = t} \F^{(t)}[v]$
    \STATE invoke the procedure \textsf{WalkPro}$(G, L, w^{(t)})$ with some minor modifications
    \STATE add $w^{(t)}$ into the seed set $S$
    \FOR{each $x \gets 0$ to $L$}
        \STATE $c^{(t)}_x \gets 0$
        \FOR{each vertex $v \in N_{G}^{O}(w^{(t)}) - S$}
            \STATE $\F^{(t)}_{x}[v] \gets \textsf{LazyF}(t, x, v)$
            \STATE $\Delta \F^{(t)}_{x}[v] \gets \textsf{LazyDF}(t, x, v)$
            \STATE $c^{(t)}_x \gets c^{(t)}_x  - P_{G}(w^{(t)}, v)(\Delta \F^{(t)}_{x}[v] + \F^{(t)}_{x}[v])$
        \ENDFOR
                \IF{$x = 0$}
            \STATE $g^{(t)}_x \gets c^{(t)}_x$
        \ELSE
            \STATE $g^{(t)}_x \gets g^{(t)}_{x - 1} + c^{(t)}_x$
        \ENDIF
    \ENDFOR
    \STATE store $c^{(t)}_x$ and $g^{(t)}_x$ in the main memory for each $0 \leq x \leq L$
    \STATE $f \gets 0$
    \FOR{each vertex $u \in V(G) - S$}
        \STATE get $t_u$ and $\F^{(t_u)}[u]$
        \IF{$\F^{(t_u)}[u] > f$}
            \STATE $y \gets t_u$
            \WHILE{$y \leq t$}
                \STATE $\F^{(y+1)}[u]  \gets \F^{(y)}[u] + \sum_{j = 1}^{L} (g^{(t)}_{L - j - 1} \cdot {\A^{(y)}}^{j}[u, w^{(y)}] )$
                \STATE $t_u \gets y + 1$
                \IF{$\F^{(y+1)}[u] \leq f$}
                    \STATE \textbf{break}
                \ENDIF
                \STATE $y \gets y + 1$
            \ENDWHILE
            \IF{$y=t$}
            \IF{$f < \F^{(t + 1)}[u]$}
                \STATE $f \gets \F^{(t + 1)}[u]$
            \ENDIF
            \ENDIF
        \ENDIF
    \ENDFOR
    \end{algorithmic}
    }
    }}}
    \label{Fig: IncInfOpt}
    \vspace{-3em}
\end{figure}

\myskip
\noindent{\textbf{\underline{Lazy Score Updating.}}}
We now describe the lazy score updating procedure \textsf{ScoreUpd-Lazy}.
To realize lazy update, each vertex $v$ in $G$ is associated with a time stamp $t_{v} \in \{1, 2, \dots, k\}$, which indicates that the score of $v$ has been updated to $F^{(t_v)}[v]$.
Before executing the $t$-th iteration, we must have $t_{v} < t$.
During the $t$-th iteration, let $f$ be a lower bound of $\F^{(t + 1)}[w^{(t+1)}]$, where $w^{(t+1)}$ is the vertex with the maximum updated score among all vertices remaining in $G$
selected in the next $(t+1)$-th iteration.
If $\F^{(t_v)}[v] \leq f$ for a vertex $v$, we need not to update the score of $v$
because $\F^{(t + 1)}[v] \leq \F^{(t_v)}[v] \le f$, and $v$ cannot be selected as a seed in the $(t+1)$-th iteration.

The \textsf{ScoreUpd-Lazy} procedure works as follows:
Let $t$ be the current iteration number.
First, we select as a seed the vertex $w^{(t)}$ with the maximum score (line~1).
Then, we obtain the column ${\A^{{(t)}}}^j[*, w^{(t)}]$ for each $1 \leq j \leq L$ by calling the \textsf{WalkPro} procedure (line~2).
Then, we add $w^{(t)}$ to the seed set $S$ (line~3).

The procedure then computes $c^{(t)}_{x}$ and $g^{(t)}_{x}$ for all $0 \leq x \leq L$ (lines~4--14).
In each loop, it first sets $c^{(t)}_{x} = 0$ (line~5).
By Eq.~\eqref{Eqn: cxt}, we need to know $\F^{(t)}_{j}[v]$ and $\Delta \F^{(t)}_{j}[v]$ for each vertex $v \in N^{O}_{G}(w^{(t)}) - S$.
Particularly, \textsf{ScoreUpd-Lazy} computes $\F^{(t)}_{j}[v]$ and $\Delta \F^{(t)}_{j}[v]$ by calling \textsf{LazyF} and \textsf{lazyDF}, respectively (lines~7--8).
The value $- P_{G}(w^{(t)}, v)(\Delta \F^{(t)}_{x}[v] + \F^{(t)}_{x}[v])$ is then added to $c^{(t)}_{x}$ (line~9).
After obtaining $c^{(t)}_{x}$, it sets $g^{(t)}_{x} = c^{(t)}_{x}$ if $x = 0$, and sets $g^{(t)}_{x} = g^{(t)}_{x-1} + c^{(t)}_{x}$ otherwise (lines~10--13).
After obtaining all $c^{(t)}_{x}$ and $g^{(t)}_{x}$, we store these values in the main memory (line~14).
At this time, we are ready update the scores of the vertices remaining in $G$.

Let $f$ represent a lower bound on $\F^{(t + 1)}[w^{(t+1)}]$. We initialize $f = 0$ at the beginning (line~15) and check each vertex $u$ remaining in $V(G) - S$ one by one (lines~16--28).
For $u$, we retrieve $t_u$ and $\F^{(t_u)}[u]$ (line~17).
If $\F^{(t_u)}[t_u] \le f$, $u$ must not have the maximum score, so we need not to update its score.
If $\F^{(t_u)}[t_u] > f$, we gradually update the score of $u$ (lines~20--25) as follows:

First, let $y = t_{u}$ (line~19). We iterate from $y = t_{u}$ until $y = t$.
Each time when we update $\F^{(y)}[u]$ to $\F^{(y+1)}[u]$, we add the value $\sum_{j = 1}^{L} g^{(y)}_{L - j - 1} {\A^{(y)}}^{j}[u, w^{(y)}]$ to $\F^{(y)}[u]$ by Eq.~\eqref{Eqn: Ft1Ft} (line~21) and update $t_u$ to $y+1$ (line~22). Since $y \leq t$, all the elements $g^{(y)}_x$ for $1 \leq x \leq L$ and the column ${\A^{{(y)}}}^j[*, w^{(y)}]$ for each $1 \leq j \leq L$ must have been computed and stored in the main memory in previous iterations. If $\F^{(y+1)}[u] \le f$ (line~23), we need not to update the score of $u$ any more and terminate lazy updating (line~24).
When $y = t$, the score of $u$ is updated to $\F^{(t+1)}[u]$. At this time, we update the lower bound $f$ to be $\F^{(t + 1)}[u]$ if $f < \F^{(t + 1)}[u]$ (lines~27--28).

After examining all vertices in $V(G) - S$, the score of the vertex $w^{(t+1)}$ with the maximum score in the $(t+1)$-th iteration must be updated to $\F^{(t + 1)}[w^{(t+1)}]$.

\myskip
\noindent{\textbf{\underline{Complexity Analysis.}}}
If we use \textsf{ScoreUpd-Lazy} in all $k$ iterations, we only need to compute $\F^{(t)}_j[v]$ and $\Delta \F^{(t)}_j[v]$ for all $1 \leq t \leq k$ and $1 \leq j \leq L$ for all vertices $v$ such that $(w^{(t')}, v) \in E(G)$ for some $1 \leq t' \leq k$. Meanwhile, at any iteration of the algorithm, for each $j$ and $v$, we store $\F^{(t)}_j[v]$ and $\Delta \F^{(t)}_j[v]$ for at most one $t$ where $1 \leq t \leq k$. Let $n' = |\bigcup_{t} N^{O}_{G}(w^{(t)})|$.
It takes at most $O(L)$ time to compute each $\Delta \F^{(t)}_j[v]$ by \textsf{LazyDF}, so the total time for computing all $\Delta \F^{(t)}_j[v]$ is $O(kL^{2}n')$.
It takes at most $O(k)$ time to compute a specific $\F^{(t)}_j[v]$, so the total time cost for computing all $\F^{(t)}_j[v]$ is $O(kLn')$.
The total space required to store all of them is $O(Ln')$.

Assume that all of $\F^{(t)}_j[v]$ and $\Delta \F^{(t)}_j[v]$ have been computed and stored, the time for computing all $c^{(t)}_{x}$ and $g^{(t)}_{x}$ in each iteration is $O(Ln)$.
The space to store them is $O(L)$ in each iteration.
The total time and space cost across all of the $k$ iterations are at most $O(kLn)$ and $O(kL)$, respectively.

In the main loop of vertex examination, we have to update $\F^{(t)}[u]$ for all vertices $u$ in the worst-case. The time cost of the each updating (line~21) is $O(L)$.
As a result, the total time cost for score updating in all $k$ iterations is $O(kLn)$.
Meanwhile, we need $O(n)$ space to store $t_u$ and $\F^{(t_u)}[u]$ for each vertex $u$.

The expected and the worst-case time complexities of \textsf{WalkPro} are $O(\Delta^{L})$ and $O(Lm)$ in each iteration, respectively.
Meanwhile, the expected and worst-case space cost of \textsf{WalkPro} is $O(\Delta^{L})$ and $O(Ln)$ in each iteration, respectively.
Since we need to reserve the column ${\A^{{(t)}}}^j[*, w^{(t)}]$ for each $1 \leq j \leq L$ and $1 \leq t \leq k$ in the main memory.
The expected and the worst-case total space costs of \textsf{WalkPro} are $O(k\Delta^{L})$ and $O(kLn)$, respectively.

Putting them together, the expected and the worst-case time complexities of \textsf{ScoreUpd-Lazy} across all of the $k$ iterations are also $O(kL^{2}n' + kLn + k\Delta^{L})$ and $O(kL^{2}n' + kLm)$, respectively.
The total expected and worst-case space cost are $O(k \Delta^L + Ln')$ and $O(kLn)$, respectively.
Although \textsf{ScoreUpd-Lazy} has the same time complexity as \textsf{ScoreUpd-Basic}, it practically runs much faster than \textsf{ScoreUpd-Basic} since we can skip the updating of lots of vertices in each iteration.

\section{The Q\MakeLowercase{uick}IM Algorithm}
\label{Sec: QIMAlg}

In this section, we present the \textsf{QuickIM} algorithm, which is the first IM algorithm that attains high time efficiency, high result quality, low memory footprint, and high robustness at the same time.
\textsf{QuickIM} follows the general greedy IM framework. It takes as input a graph $G$, the desired number of seeds $k$ and the longest walk length $L$. At the beginning, the seed set $S$ is set to be empty (line~1).
The influence scores of all vertices in $G$ are estimated by the \textsf{ScoreEst} procedure (line~2).
We initialize the time stamp $t_v$ as $1$ for each vertex $v \in V(G)$ (line~3).
Then, the algorithm carries out $k$ iterations.
In each iteration, it calls the \textsf{ScoreUpd-Lazy} procedure to find a seed and update the scores of other vertices (line~5). Finally, the seed set $S$ is outputted (line~6).

\begin{figure}[!t]
    \centering
    \scriptsize
    \resizebox{\algswidth}{!}{
    \fbox{
    \parbox{\figwidth}{
    {
    \textbf{Algorithm} \textsf{QuickIM}$(G, k, L)$
    \begin{algorithmic}[1]
    \STATE $S \gets \emptyset$
	\STATE \textsf{ScoreEst}$(G, L)$
    \STATE $t_{v} \gets 1$ for each vertex $v \in V(G)$
	\FOR{$t \gets 1$ to $k$}
		\STATE \textsf{ScoreUpd-Lazy}$(G, t, L, S)$
	\ENDFOR
	\RETURN $S$
    \end{algorithmic}
    }
    }}}
    \label{Fig: QuickIM}
    \vspace{-2em}
\end{figure}

\myskip
\noindent{\textbf{\underline{Evaluation Results.}}} We report our evaluation results of \textsf{QuickIM} as follows.

\begin{itemize}
\item{\textit{\underline{Time Efficiency:}}}
The expected and the worst-case time complexities of \textsf{QuickIM} are
$O(Lm + kLn + kL^{2}n' + k\Delta^{L})$ and $O(kLm + kL^{2}n')$, respectively.
Here, $L$ is often set to be a very small number in our algorithm, $\Delta$ is also small for a real social network,
and $n'$ is a small number far less than $n$.
Thus, the expected time complexity of \textsf{QuickIM} is $O(m + kn)$.
When $k$ is fixed, \textsf{QuickIM} attains the $\Omega(m+n)$ lower bound of the time complexity of an IM algorithm~\cite{Borgs2012Maximizing}.
As shown in Table~\ref{Tab: IMAlgCmp}, both the expected and the worst-case time complexity of \textsf{QuickIM} are the lowest among all the algorithms.
As verified by the experimental results in Section~\ref{Sec: PerEva}, \textsf{QuickIM} runs $1$--$3$ orders of magnitude faster than all the other IM algorithms.

\item{\textit{\underline{Result Quality:}}}
As \textsf{QuickIM} follows a similar framework of the \textsf{EasyIM} algorithm, it also produces high quality results in practice. As verified by the experimental results in Section~\ref{Sec: PerEva}, the result quality of \textsf{QuickIM} is comparable to that of \textsf{EasyIM} and other start-of-the-art IM algorithms.

\item{\textit{\underline{Memory Footprint:}}}
The expected and worst-case memory cost of \textsf{QuickIM} are $O(Ln + k\Delta^{L})$ and $O(kLn)$, respectively.
In practice, the memory overhead of \textsf{QuickIM} is very low.
This is because the parameter $k$ is often set to a small parameter in real-world IM applications, $\Delta$ is often small for real-world social networks
and it is actually sufficient to set $L$ to be a very small number.
At this time,  the memory cost of \textsf{QuickIM} attains $\Omega(n)$, which is linear w.r.t.~the number of vertices.
As verified by the experimental results in Section~\ref{Sec: PerEva}, the memory cost of \textsf{QuickIM} is comparable to \textsf{EasyIM} and $1$--$2$ orders of magnitude less than all the other IM algorithms.

\item{\textit{\underline{Influence Robustness:}}}
\textsf{QuickIM} is insensitive to influence probabilities.
This is because the score estimation and the score updating procedures only perform numerical operations on influence probabilities. Thus, the time complexity of \textsf{QuickIM} is not affected by the values of influence probabilities.
As verified in Section~\ref{Sec: PerEva}, the running time of \textsf{QuickIM} is very stable for various influence probability settings.
\end{itemize}

\section{Performance Evaluation}
\label{Sec: PerEva}

We conducted extensive experiments to evaluate the \textsf{QuickIM} algorithm. The experimental results are reported in this section.

\subsection{Setup}
\label{Sec: PerEva-1}

\noindent{\textbf{\underline{Algorithms.}}}
We implemented \textsf{QuickIM} in C++.
The implementation is available at \url{https://github.com/Akaisorani/QuickIM}.
For comparisons, we choose some state-of-the-art IM algorithms as competitors:
1) In the category of simulation-based algorithms, we choose \textsf{PrunedMC}~\cite{Ohsaka2014Fast}. This algorithm has been shown to be much faster than
\textsf{Greedy}, \textsf{CELF}, and \textsf{CELF++} in~\cite{Ohsaka2014Fast}.
It has very similar performance as \textsf{StaticGreedy}~\cite{Cheng2013StaticGreedy} in terms of time efficiency and result quality~\cite{Arora2017Debunking}.
2) Among the reverse sampling algorithms, we choose three algorithms \textsf{D-SSA}~\cite{Nguyen2016Stop}, \textsf{SKIS}~\cite{Nguyen2017Importance} and \textsf{Coarsen}~\cite{Ohsaka2017Coarsening} because they run much faster than the other reverse sampling algorithms \textsf{RIS}~\cite{Borgs2012Maximizing}, \textsf{TIM/TIM+}~\cite{Tang2014Influence} and \textsf{IMM}~\cite{Tang2015Influence} as reported in~\cite{Nguyen2017Importance, Nguyen2016Stop}.
3) In the category of score estimation algorithms, \textsf{EasyIM} is chosen.
As reported in~\cite{Arora2017Debunking}, the result quality of \textsf{IRIE} is rather low in comparison with \textsf{EasyIM}.

\begin{table}[t]
    \centering
    \scriptsize
    \vspace{-1em}
    \caption{Statistics of Networks Used in Experiments.}
    \resizebox{\columnwidth}{!}{
    \begin{tabular}{lrrr}
    	\hline
        \rowcolor{mygray}
    	Network name& \# of vertices & \# of edges & Avg. degree\\
    	\hline
    	\textsl{DBLP} & 654,628 & 3,980,318 & 6.08\\
        \textsl{YouTube} & 1,134,890 & 5,975,248 & 5.27\\
         \textsl{LiveJournal} & 4,847,571 & 68,993,773 & 14.23\\
         \textsl{Orkut} & 3,072,441 & 234,370,166 & 76.28\\
         \textsl{Twitter} & 61,578,415 & 1,468,364,884 & 23.85\\
         \textsl{Friendster} & 65,608,366 & 3,612,134,270 & 55.06\\
        \hline
    \end{tabular}
    }
    \label{Tab: Datasets}
    \vspace{-2em}
\end{table}
\begin{figure*}[t]
\centering
\includegraphics[width= \textwidth]{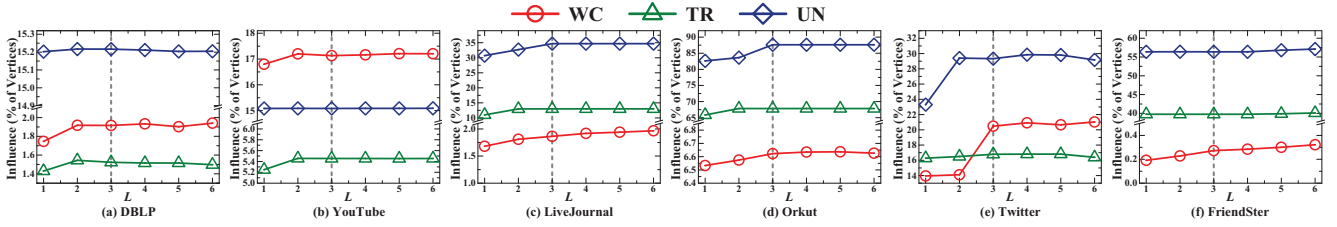}
%
\vspace{-1em}
\caption{Effects of Parameter \emph{L} on Result Quality}
\label{Fig: Exp: LT}
\vspace{-1em}
\end{figure*}

\myskip
\noindent{\textbf{\underline{Datasets.}}}
We tested all these algorithms on the datasets that have been widely used in the evaluation of IM algorithms~\cite{Nguyen2017Importance, Nguyen2016Stop, Ohsaka2017Coarsening, Tang2015Influence, Tang2014Influence}.
In particular, we use six large real-world social networks taken from the arXiv\footnote{\url{https://arxiv.org/}} and SNAP\footnote{\url{http://snap.stanford.edu/}} repositories.
Table~\ref{Tab: Datasets} summarizes the statistics of these networks. The largest networks \textsl{Twitter} and \textsl{FriendStar} contain billions of edges.

\myskip
\noindent{\textbf{\underline{Influence Probability Assignment.}}}
We assign influence probabilities $P_{G}(u, v)$ to edges $(u, v)$ according to three widely adopted models~\cite{Arora2017Debunking, Nguyen2017Importance, Nguyen2016Stop, Ohsaka2017Coarsening}:
1) In the \emph{weighted cascade (\textsl{WC}) model}, $P_{G}(u, v) = 1 /d_{G}^{I}(v)$, where $d_{G}^{I}(v)$ is the in-degree of $v$.
2) In the \emph{trivalency (\textsl{TR}) model}, $P_{G}(u, v)$ is chosen uniformly at random from three numbers $p_t, p_t^2, p_t^3 \in (0, 1)$, where $p_r = 0.1$ by default.
3) In the \emph{uniform (\textsl{UN}) model}, $P_{G}(u, v) = p_u \in (0, 1)$, where $p_u = 0.1$ by default.

\begin{figure*}[t]
\includegraphics[width= \textwidth]{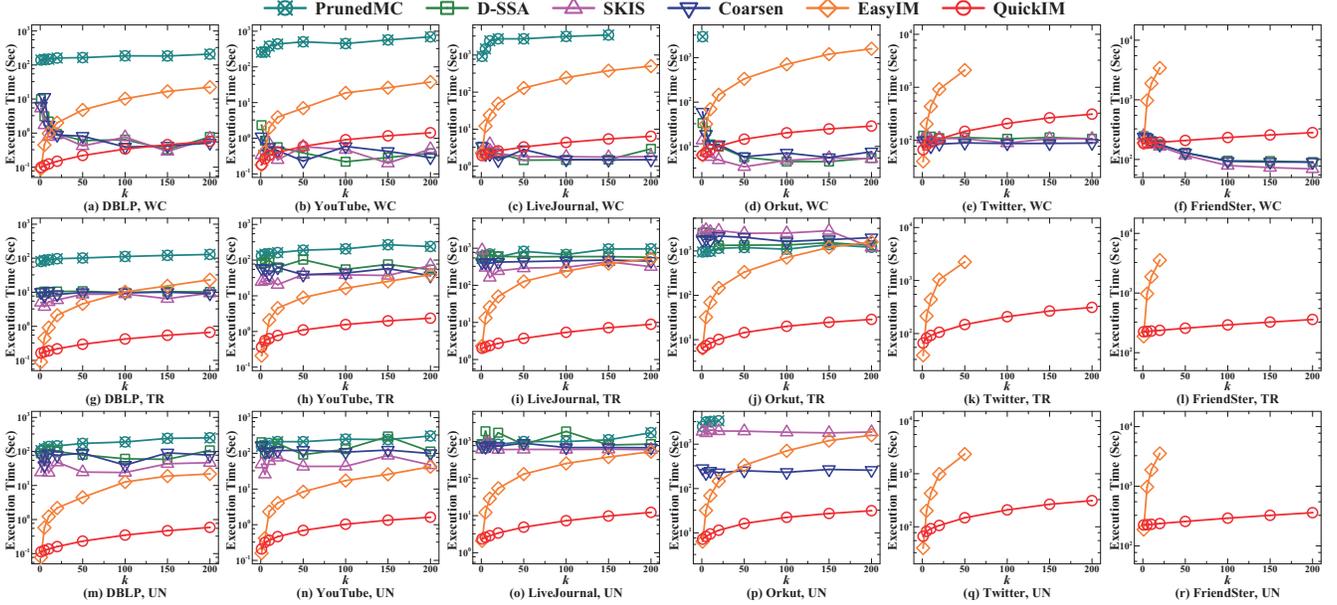}
%
%
%
\vspace{-1em}
\caption{Execution Time of IM Algorithms with respect to Number of Seeds \emph{k}.}
\label{Fig: Exp: T}
\vspace{-1.2em}
\end{figure*}
\begin{table*}
	\caption{Result Quality of IM Algorithms.}
    \vspace{0.2em}
    \resizebox{\textwidth}{!}
    {
    \begin{tabular}{c|c|rrrrrrr||rrrrrrr}
    	\hline
        \rowcolor{mygray}
        {\bf Probability } &  & & \multicolumn{5}{c}{\textbf{Influence (\% of vertices) when {\it k} = 50}} & & & \multicolumn{5}{c}{\textbf{Influence (\% of vertices) when {\it k} = 100}} & \\
        \cline{3-14}
        \rowcolor{mygray}
         {\bf Model} & {\bf Network} & {\bf \textsf{PrunedMC}}   & {\bf \textsf{D-SSA}}   & {\bf \textsf{SKIS}} & {\bf \textsf{Coarsen}} & {\bf \textsf{EasyIM}} & {\bf \textsf{QuickIM}}  & {\bf Diff}   & {\bf \textsf{PrunedMC}}   & {\bf \textsf{D-SSA}}   & {\bf \textsf{SKIS}} & {\bf \textsf{Coarsen}} & {\bf \textsf{EasyIM}} & {\bf \textsf{QuickIM}}  & {\bf Diff}  \\ \hline

        \multirow{6}{*}{\textsl{WC}} &
        \textsl{DBLP} & \textbf{1.173} & $0.980$ & $0.968$ & $1.107$ & $1.102$ & $1.128$ & {0.045} & {\bf 1.933} & $1.358$ & $1.336$ & $1.557$ & $1.903$ & $1.916$ & {0.017}\\
        & \textsl{YouTube} & {\bf 13.075} & $11.545$ & $11.931$ & $12.042$ & $13.029$ & $13.033$ & {0.042} & {\bf 17.362} & $14.844$ & $13.650$ & $15.314$ & $17.018$ & $17.132$ & {0.230}\\
        & \textsl{LiveJournal} & {\bf 1.573} & $1.333$ & $1.312$ & $1.372$ & $1.558$ & $1.436$ & {0.137} & {\bf 2.015} & $1.774$ & $1.598$ & $1.648$ & $1.788$ & $1.864$ & {0.151}\\
        & \textsl{Orkut} & N/A & $4.619$ & $4.042$ & $4.548$ & $5.263$ & {\bf 5.335}  & {0} & N/A & $4.957$ & $5.173$ & $5.055$ & $6.602$ & {\bf 6.623} & {0}\\
        & \textsl{Twitter} & N/A & $16.217$ & $15.612$ & $15.393$ & {\bf 16.232} & $16.158$  & {0.074} & N/A & $20.083$ & {\bf 20.973} & $20.174$ & N/A & $20.478$ & {0.495}\\
        & \textsl{Firendster} & N/A & {\bf 0.113} & $0.098$ & $0.092$ & N/A & $0.107$ & {0.006} & N/A & {\bf 0.298} & $0.288$ & $0.251$ & N/A & $0.273$ & {0.025} \\ \cline{1-16}

        \multirow{6}{*}{\textsl{TR}}
        & \textsl{DBLP} & {\bf 1.652} & $1.546$ & $1.556$ & $1.508$  & $1.601$ & $1.487$ & {0.165} & {\bf 1.730} & $1.584$ & $1.566$ & $1.559$& $1.691$ & $1.523$ & {0.207}\\
        & \textsl{YouTube} & {\bf 5.490} & $5.449$ & $5.468$ & $5.257$ & $5.469$ & $5.448$ & {0.042}&  5.517 & $5.469$ & $5.462$ & $5.265$ & {\bf 5.529} & $5.455$ & {0.074}\\
        & \textsl{LiveJournal} & $13.097$ & $12.993$ & $13.002$ & $12.208$ & {\bf 13.108} & $12.991$ & {0.117} & $13.147$ & $12.997$ & $13.276$ & $12.228$ & {\bf 13.219} & $12.996$ & {0.280}\\
        & \textsl{Orkut} & $67.840$ & $67.800$ & {\bf 67.841} & $60.452$ & $67.686$ & {\bf 67.841} & {0} &  $67.846$ & $67.843$ & $67.846$ & $60.448$ & $67.822$ & {\bf 67.878} & {0}\\
        & \textsl{Twitter} & N/A & N/A & N/A & N/A & {\bf 16.640} & $16.239$ & {0.401} & N/A & N/A & N/A & N/A & N/A & {\bf 16.794} & {0}\\
        & \textsl{Firendster} & N/A & N/A & N/A & N/A & N/A & {\bf 39.219} & {0} & N/A & N/A & N/A & N/A & N/A &{\bf 39.770}& {0}\\ \cline{1-16}

        \multirow{6}{*}{\textsl{UN}}
        & \textsl{DBLP} & {\bf 15.238} & $15.187$ & $15.202$ & $14.611$ & $15.042$ & $15.197$ & {0.041} & {\bf 15.361} & $15.193$ & $15.292$ & $14.639$ & $15.312$ & $15.217$ & {0.144} \\
        & \textsl{YouTube} & {\bf 15.123} & $15.003$ & $15.101$ & $14.301$ & $15.023$ & $15.086$ & {0.037} & $15.252$ & {\bf 15.488} & $15.202$ & $14.798$ & $15.211$ & $15.089$ & {0.399}\\
        & \textsl{LiveJournal} & $34.095$ & $34.108$ & $34.010$ & $30.012$ & $33.959$ & {\bf 34.702} & {0} & {\bf 34.840} & $34.697$ & $34.713$ & $30.167$ & $34.815$ & $34.703$ & {0.137}\\
        & \textsl{Orkut} & N/A & N/A & {\bf 87.630} & $72.723$ & $87.611$ & $87.628$ & {0.002} &  N/A & N/A & {\bf 87.633} & $72.733$ & $87.615$ & $87.630$ & {0.003}\\
        & \textsl{Twitter} & N/A & N/A & N/A & N/A & {\bf 29.454} & 29.101 & {0.353} &  N/A & N/A & N/A & N/A & N/A & {\bf 29.331} & {0}\\
        & \textsl{Firendster} & N/A & N/A & N/A & N/A & N/A & {\bf 55.108} &  {0} & N/A & N/A & N/A & N/A & N/A & {\bf 56.386} & {0}\\ \hline

\end{tabular}
    }
    \label{Tab: Exp-I}
    \vspace{-1.2em}
\end{table*}
\begin{figure*}[t]
\includegraphics[width= \textwidth]{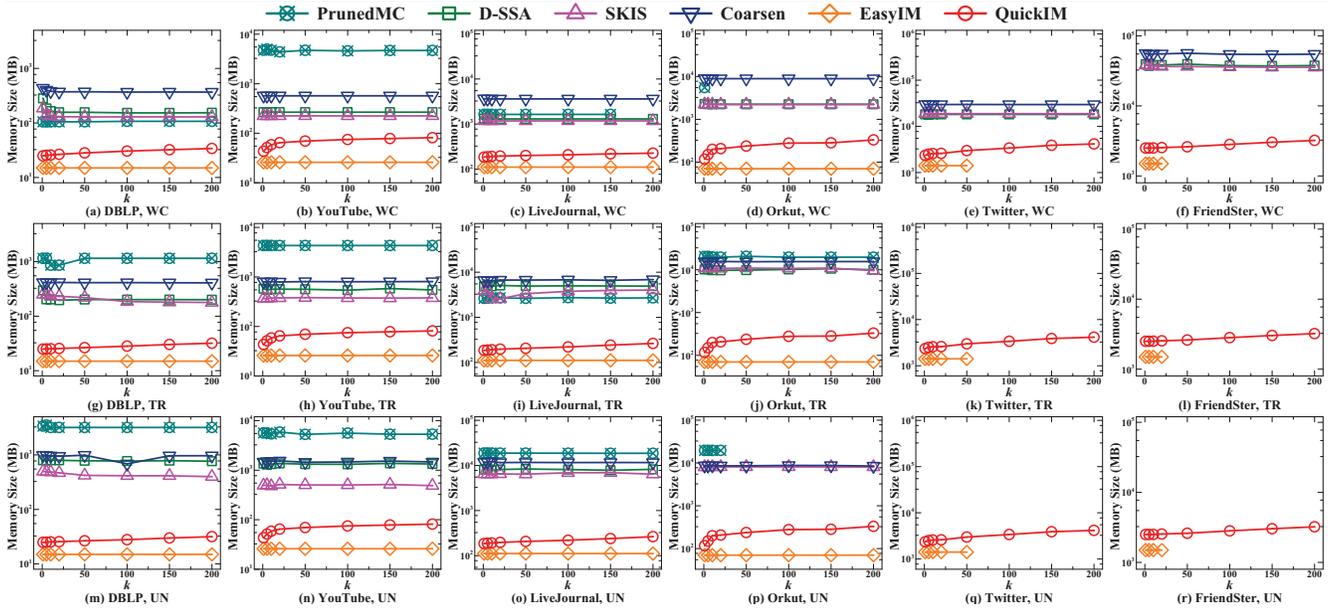}
%
%
%
\vspace{-1em}
\caption{Memory Footprint of IM Algorithms with respect to Number of Seeds \emph{k}.}
\label{Fig: EXP: M}
\vspace{-0.5em}
\end{figure*}
\begin{figure*}[t]
\includegraphics[width= \textwidth]{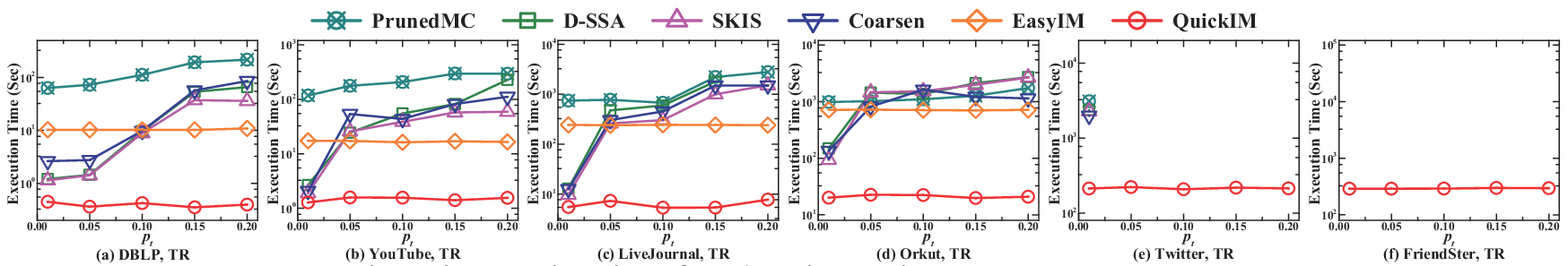}
%
\vspace{-1em}
\caption{Execution Time of IM Algorithms with respect to Parameter $p_{t}$.}
\label{Fig: EXP: RTR}
\vspace{-0.5em}
\end{figure*}
\begin{figure*}[!t]
\includegraphics[width= \textwidth]{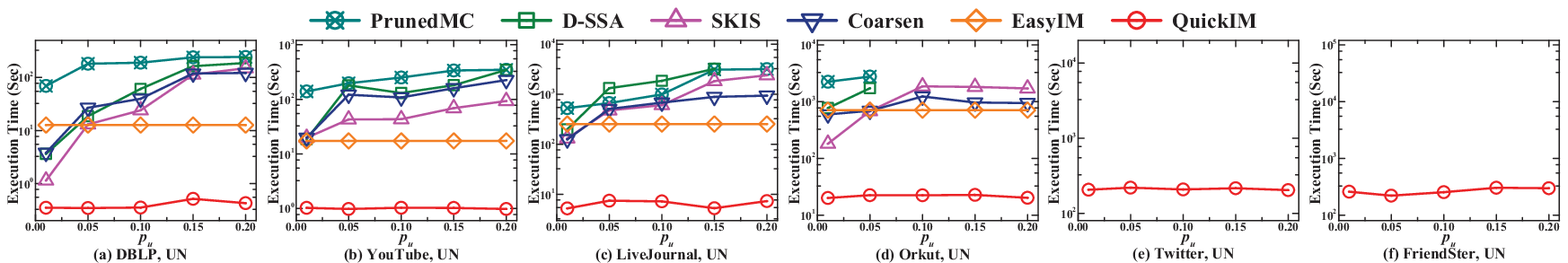}
%
\vspace{-1em}
\caption{Execution Time of IM Algorithms with respect to Parameter $p_{u}$.}
\label{Fig: EXP: RUN}
\vspace{-1.3em}
\end{figure*}
\myskip
\noindent{\textbf{\underline{Parameter Settings.}}}
For the sake of fairness, we set the parameters of the tested algorithms to their default values.
In particular, for \textsf{PrunedMC}, we set the sample size $\theta_{s}$ to 200 as recommended in~\cite{Arora2017Debunking}.
For \textsf{D-SSA}, \textsf{SKIS} and \textsf{Coarsen}, we set the error bound $\epsilon$ to 0.5.
For \textsf{SKIS}, the parameter $h$ for controlling the sample size is set to 5. For \textsf{Coarsen}, the iteration number $r$ for coarsening the input graph is set to 16.
For \textsf{EasyIM}, the maximum length $L$ of paths is set to 3.
For our \textsf{QuickIM} algorithm, we also set the maximum length $L$ of walks to 3 by default.
The reasons for such parameter settings will be elaborated in Section~\ref{Sec: PerEva-2}.

\myskip
\noindent{\textbf{\underline{Metrics.}}}
We evaluate the performance of the algorithms by three metrics:
1) \emph{Execution time}: For each algorithm, we only examine the online execution time to find $k$ seeds.
In other words, we do not account for the index construction time of \textsf{SKIS}, the graph coarsening time of \textsf{Coarsen} and the Monte-Carlo estimation time of \textsf{EasyIM}.
Note that we terminate the execution of an algorithm if it is unable to finish in an hour.
2) \emph{Memory footprint}: We examine the amount of main memory for storing all auxiliary data structures used by the algorithms.
Hence, the index size is counted in the memory footprint of \textsf{SKIS}, and the coarsened graph size is counted in the memory footprint of \textsf{Coarsen}.
3) \emph{Result quality}: The quality of a seed set is measured by the expected fraction of vertices that can be influenced by the seeds. It is \#P-hard to compute this quality measure.
Following the prior work~\cite{Arora2017Debunking, Galhotra2016Holistic, Ohsaka2017Coarsening}, we conduct $10^{4}$ Monte-Carlo simulations and compute the average fraction of vertices influenced by the seeds in the $10^4$ simulations\footnote{For \textsf{Coarsen}, the simulations are done on the coarsened graph.}.

All the experiments were performed on a machine with an Intel Xeon CPU (2.2GHz, 16 cores), 512GB of DDR4 RAM and 1.2TB of SAS disks, running CentOS 7.

\subsection{Experimental Results}
\label{Sec: PerEva-2}

\noindent{\textbf{\underline{Effects of Parameter \emph{L}.}}}
First, we examined the effects of the parameter $L$ of \textsf{QuickIM}. The goal is to determine a proper value of $L$ to ensure the result quality of \textsf{QuickIM}.
In the experiment, we set $k = 100$ and vary $L$ from $1$ to $6$.
Fig.~\ref{Fig: Exp: LT} illustrates the percentage of influenced vertices w.r.t.~$L$. For all of the influence probability assignment models \textsl{WC}, \textsl{TR} and \textsl{UN}, we all have the following observations:
1) More vertices are influenced as $L$ grows. This is because when $L$ is larger, more walks are involved in influence score estimation.
As a result, the estimated score of a vertex with higher influence generally increases more significantly than a vertex with low influence.
Therefore, \textsf{QuickIM} is more likely to select seeds with really high influence, thereby improving the result quality.
2) When $L \geq 3$, the improvement in result quality is diminishing.
The reason is that the probability of a walk decreases exponentially as $L$ grows according to Lemma~\ref{Lem: WalkPr}.
Hence, when $L$ is sufficiently large, the estimated influence scores of all vertices tend to be unchanged, so \textsf{QuickIM} tends to select the same set of seeds.

We also examined the effects of $L$ on the execution time and the memory overhead of \textsf{QuickIM}. We find that both of them grow exponentially w.r.t.~$L$.
Due to space limits, we show the detailed experimental results in Appendix~C of the full paper~\cite{Zhu2018FullVersion}.
Because of the observations above, we use $L = 3$ as the default value.

\myskip
\noindent{\textbf{\underline{Time Efficiency.}}}
In this experiment, we tested the execution time of the IM algorithms by varying $k$ from $1$ to $200$.
The results are illustrated in Fig.~\ref{Fig: Exp: T}.
We have the following observations:

1) \textsf{QuickIM} is able to find a set of good seeds in up to 4 minutes on all the networks. However, any other algorithms may fail to terminate in an hour for sufficiently large $k$ on large networks, especially on the larger networks \textsl{Twitter} and \textsl{Friendster}.

2) \textsf{QuickIM} runs 1--3 orders of magnitude faster than the state-of-the-art simulation-based algorithm \textsf{PrunedMC}. On the larger networks \textsl{Orkut}, \textsl{Twitter} and \textsl{FriendStar}, \textsf{PrunedMC} is often unable to finish in an hour.
This is simply because the time complexity of \textsf{PrunedMC} is much higher than \textsf{QuickIM}.

3) When influence probabilities are assigned according to the \textsl{WC} model, the state-of-the-art reverse sampling algorithms \textsf{D-SSA}, \textsf{SKIS} and \textsf{Coarsen} are usually 3$\times$--5$\times$ faster than \textsf{QuickIM}.
However, when the \textsl{TR} and \textsl{UN} models are used, \textsf{QuickIM} in turn runs 1--3 orders of magnitude faster than \textsf{D-SSA}, \textsf{SKIS} and \textsf{Coarsen}.
The reasons are as follows:
\begin{itemize}
\item According to the \textsl{WC} model, the expected number of in-edges incident to every vertex is $1$, so the size of each sample in the reverse sampling algorithms is often very small~\cite{Nguyen2017Importance}.
\item When the \textsl{TR} and \textsl{UN} models are used, the reverse sampling algorithms must sample more edges, thereby consuming more time. As analyzed in Section~\ref{Sec: RevAlg-2}, their time complexities grow exponentially with respect to influence probabilities.
\item The execution time of \textsf{QuickIM} is independent of influence probabilities because \textsf{QuickIM} only carries out graph traversal and arithmetic computations. Fig.~\ref{Fig: Exp: T} verifies this.
\end{itemize}

4) Both \textsf{QuickIM} and \textsf{EasyIM} are score estimation algorithms. Their execution time follow similar trends. However, \textsf{QuickIM} is 1--2 orders of magnitude faster than \textsf{EasyIM}.
This is because \textsf{EasyIM} has to scan the entire graph $L$ times whenever scores are updated, whereas \textsf{QuickIM} updates scores incrementally and only accesses a small portion of vertices.

\myskip
\noindent{\textbf{\underline{Result Quality.}}}
In this experiment, we compare \textsf{QuickIM} with the other IM algorithms in terms of result quality.
In Table~\ref{Tab: Exp-I}, we list the expected fraction of vertices influenced by the discovered seeds.
A few results are unavailable (marked by ``N/A'') because those algorithms are unable to terminate in an hour.
The column entitled ``\emph{Diff}'' records the  difference between the result quality of \textsf{QuickIM} and the best of all the algorithms (highlighted in bold).
The absolute difference is less than 0.5\% in all cases  and 0.1\% in most cases (0 means that the result of \textsf{QuickIM} is the best).
From Table~\ref{Tab: Exp-I}, we find that \textsf{QuickIM} can yield as high quality results as the best algorithms no matter which models \textsl{WC}, \textsl{TR} or \textsl{UN} are used to assign influence probabilities.
This verifies that the influence score function and the score estimation and updating methods used by \textsf{QuickIM} are effective.

\myskip
\noindent{\textbf{\underline{Memory Footprint.}}}
In this experiment, we examine the memory footprint of the IM algorithms by varying $k$ from $1$ to $200$.
The results are shown in Fig.~\ref{Fig: EXP: M}. We have the following observations:

1) \textsf{EasyIM} is the most memory-efficient. It requires less memory than \textsf{QuickIM} by a factor of 2--5. This is because the space complexity of \textsf{EasyIM} is only $O(n)$.

2) Except \textsf{EasyIM}, \textsf{QuickIM} requires 1--2 orders of magnitude less memory than all the other algorithms. For example, \textsf{QuickIM} requires less than 3GB of memory to handle the largest network \textsl{Friendster}, which contains more than 3.6 billion edges. However, as reported in~\cite{Nguyen2017Importance}, \textsf{SKIS} needs about 99GB of memory to find 100 seeds on \textsl{Friendster} in several hours under the \textsl{TR} model.
This is because the simulation-based algorithms and the reverse sampling algorithms must store a large number of samples for seed selection, which is extremely memory consuming.

3) The memory footprint of \textsf{QuickIM} grows linearly but slowly to parameter $k$. \textsf{QuickIM} stores four types of data in the memory:
First, it uses $O(Ln)$ memory during initial influence score estimation, which is independent of $k$.
Second, it uses $O(n)$ memory to store the vector $\F^{(t)}$, which is also independent of $k$.
Third, it uses at most $O(Ln')$ memory to store all $\F^{(t)}_i[v]$ and $\Delta \F^{(t)}_i[v]$, where $n' \ll n$, which is independent of $k$ too.
Finally, it expectedly uses $O(k\Delta^{L})$ memory to store some columns of matrix ${\A^{(t)}}^{j}$ in each iteration, where $\Delta$ and $L$ are small constants for real social networks.
Hence, $k$ has an insignificant effect on the memory footprint of \textsf{QuickIM}.

4) The memory overhead of \textsf{QuickIM} is independent of how influence probabilities are assigned.
However, the memory overheads of the other algorithms increase significantly from \textsl{WC} to \textsl{TR} and from \textsl{TR} to \textsl{UN}.
The reasons are as follows:
First, the space complexity of \textsf{QuickIM} is totally independent of influence probabilities.
Second, for the simulation-based algorithms and the reverse sampling algorithms, when influence probabilities become larger, each sample generally contains more vertices and edges, so more memory is used to store samples.

\myskip
\noindent{\textbf{\underline{Robustness.}}}
In this experiment, we further examine the robustness of the IM algorithms against influence probabilities.
Let $k = 100$. We vary the parameter $p_{t}$ in the \textsl{TR} model and the parameter $p_{u}$ in the \textsl{UN} model from 0.01 to 0.2.
The execution time of the IM algorithms with respect to $p_{t}$ and $p_{u}$ is illustrated in Fig.~\ref{Fig: EXP: RTR} and Fig.~\ref{Fig: EXP: RUN}, respectively.
We find that the execution time of \textsf{QuickIM} is independent of $p_{t}$ and $p_{u}$.
However, the execution time of all the sampling-based algorithms \textsf{PrunedMC}, \textsf{D-SSA}, \textsf{SKIS} and \textsf{Coarsen} grows exponentially as $p_{t}$ or $p_{u}$ gets larger.
The reasons have been clarified earlier. The execution time of \textsf{EasyIM} is also independent of $p_{t}$ and $p_{u}$. However, it runs much slower than \textsf{QuickIM}.
In addition, we examine the memory overhead of \textsf{QuickIM} with respect to influence probabilities.
The experimental results show that the memory overhead of \textsf{QuickIM} is also independent of $p_{T}$ and $p_{U}$. Due to space limits, we put these experimental results in Appendix~C of the full paper~\cite{Zhu2018FullVersion}.

\section{Conclusions}
\label{Sec: ConClu}

None of the existing IM algorithms satisfy all the desirable properties of a practically applicable IM algorithm, namely high time efficiency, good result quality, low memory footprint, and high robustness. \textsf{QuickIM} is the first versatile IM algorithm that satisfies all these properties at the same time. The superiority of \textsf{QuickIM} results from the score estimation IM paradigm, the walk-based influence score function, the efficient and accurate score estimation method, and the incremental score updating method.

\section{Acknowledgements}

This work was supported in part by the National Natural Science Foundation of China under grant No.~61672189, 61532015, and 61732003.
We thank the authors of \cite{Arora2017Debunking, Galhotra2016Holistic, Nguyen2017Importance, Nguyen2016Stop, Ohsaka2014Fast, Ohsaka2017Coarsening}
for sharing their source codes with us.


\bibliographystyle{abbrv}
\bibliography{quickim}
\clearpage

\clearpage

\setcounter{lemma}{0}
\setcounter{corollary}{0}
\setcounter{theorem}{0}

\appendix

\section{Proofs}

\subsection{Proof of Lemma~1}
\setcounter{equation}{2}
\begin{lemma}
\label{Lem: gLPr}
For any possible world $g^{L}$ of $G^{L}$, we have
\begin{equation}
\Pr(g^{L})  = \!\!\!\!\!\!\!\!  \prod_{e \in E(G)| \alpha_{g^{L}}(e) > 0} \!\!\!\!  {P_{G}(e)}^{\alpha_{g^{L}}(e)} \!\!\!\!\!\!\!\!   \prod_{e \in E(G)| \alpha_{g^{L}}(e) < L} \!\!\!\!\!\!\!\!   1 - P_{G}(e).
\end{equation}
\end{lemma}
\begin{proof}
For each edge $e = (u, v) \in E(G)$, there exist $\alpha_{g^{L}}(e)$ distinct edges $e^{1}, e^{2}, \dots, e^{\alpha_{g^{L}}(e)}$ from the vertex $u$ to the vertex $v$ in the possible world $g^{L}$.
If $\alpha_{g^{L}}(e) < L$, the edge $e^{\alpha_{g^{L}}(e) + 1}$ must not exist.
Therefore, if $\alpha_{g^{L}}(e) < L$, the existing probability of the $\alpha_{g^{L}}(e)$ distinct edges from $u$ to $v$ is ${P_{G}(e)}^{\alpha_{g^{L}}(e)} (1 - P_{G}(e)) $.
If $\alpha_{g^{L}}(e) = L$, the existing probability is  ${P_{G}(e)}^{\alpha_{g^{L}}(e)}$.
Since the edges in $G$ are independent, the existing probability of $g^{L}$ is
\begin{equation*}
\begin{split}
& \Pr(g^{L}) \\
& = \prod_{e \in E(G)| \alpha_{g^{L}}(e) < L}  \!\!\!\!\!\!\!\!\!\!\!\!   {P_{G}(e)}^{\alpha_{g^{L}}(e)} (1 - P_{G}(e)) \!\!\!\!\!\!\!\!   \prod_{e \in E(G)| \alpha_{g^{L}}(e) = L}  \!\!\!\!\!\!\!\!   {P_{G}(e)}^{\alpha_{g^{L}}(e)} \\
& = \prod_{e \in E(G)| \alpha_{g^{L}}(e) > 0} \!\!\!\!  {P_{G}(e)}^{\alpha_{g^{L}}(e)} \!\!\!\!\!\!\!\!   \prod_{e \in E(G)| \alpha_{g^{L}}(e) < L} \!\!\!\!\!\!\!\!   1 - P_{G}(e).
\end{split}
\end{equation*}
The lemma thus holds.
\end{proof}

\subsection{Proof of Lemma~2}
\begin{lemma}
\label{Lem: WalkPr}
For a walk $W = (v_0, v_1, \dots, v_t)$ in the graph $G$,
\begin{equation}
\Pr(W) = \prod_{i = 0}^{t-1} P_{G}(v_{i}, v_{i+1}) = \prod_{(u, v) \in W} {P_{G}(u, v)}^{\alpha_{W}(u, v)}.
\end{equation}
\end{lemma}
\begin{proof}
Let $g^L$ be a possible world of $G^L$ in which $W$ is embedded. For each edge $(u, v)$ in $W$, there must exist at least $\alpha_{W}(u, v)$ edges from the vertex $u$ to the vertex $v$ in $g^{L}$.
For $i \geq 1$, the existence of the edge $(u, v)^{i+1}$ depends on the existence of the edge $(u, v)^{i}$.
Thus, all of the edges $(u, v)^{1}, (u, v)^{2}, \dots, (u, v)^{\alpha_{W}(u, v)}$ must exist in $g^{L}$. The probability is therefore ${P_{G}(u, v)}^{\alpha_{W}(u, v)}$.
Since the edges in $G$ are independent, we have Eq.~(4). Thus, the lemma holds.
\end{proof}

\subsection{Proof of Lemma~3}
\begin{lemma}
\label{Lem: MWalkPr}
For several walks $W_1, W_2, \dots, W_t$ in the graph $G$,
\begin{equation}
\Pr(\bigwedge_{j=1}^{t}W_j) = \!\!\!\!\!\!\!\!\!\!\!\! \prod_{(u, v) \text{ is an edge in any of } W_1, W_2, \dots, W_t} \!\!\!\!\!\!\!\!\!\!\!\!\!\!\!\!\!\!\!\!\!\!\!\! {P_{G}(u, v)}^{\max_{1 \leq i \leq t}\alpha_{W_i}(u, v)}.
\end{equation}
\end{lemma}
\begin{proof}
Let $g^L$ be a possible world of $G^L$ in which all of the walks $W_1, W_2, \dots, W_t$ are embedded.
For each edge $(u, v)$ in any of $W_1, W_2, \dots, W_t$,
there must exist $\max_{1 \leq  j \leq t} \alpha_{W_j}(u, v)$ edges from the vertex $u$ to the vertex $v$ in $g^{L}$.
Similar to the proof of Lemma~2, for $i \geq 1$, the existence of the edge $(u, v)^{i+1}$ depends on the existence of the edge $(u, v)^{i}$.
Thus, all of the edges $(u, v)^{1}, (u, v)^{2}, \dots, (u, v)^{\max_{1 \leq  j \leq t} (\alpha_{W_j}(u, v))}$ must exist in $g^{L}$. The probability is therefore ${P_{G}(u, v)}^{\max_{1 \leq  j \leq t} (\alpha_{W_j}(u, v))}$. Since the edges in $G$ are independent,
we have Eq.~(5). Thus, the lemma holds.
\end{proof}

\subsection{Proof of Lemma~4}
\begin{lemma}
\label{Lem: IGPrW}
\begin{equation}\label{eqn:IGPrW}
\et
\begin{split}
~ & I_{G}(u, v) = \Pr(\bigvee_{i = 1}^{h_{uv}} W_i) = \sum_{i=1}^{h_{uv}} \Pr(W_{i}) - \!\!\!\! \!\!\!\!  \sum_{1 \leq i < j \leq h_{uv}} \!\!\!\! \!\!\!\! \Pr(\bigwedge_{i =1}^{h_{uv}} W_i) + \dots \\
& + (-1)^{t-1} \!\!\!\!\!\!\!\!  \sum_{C \subseteq [h_{uv}], |C| = t} \!\!\!\! \!\!\!\!  \Pr(\bigwedge_{i \in C} W_i) + \dots +
(-1)^{h_{uv}-1} \Pr(\bigwedge_{i = 1}^{h_{uv}} W_i).
\end{split}
\end{equation}
\end{lemma}
\begin{proof}
For each possible world $g \in \G$ and a possible world $g^{L} \in {\G}^{L}$, we say $g$ is embedded in $g^{L}$ if it satisfies:
1) $\alpha_{g^{L}}(e) \geq 1$ for each edge $e \in E(g)$; and
2) $\alpha_{g^{L}}(e) = 0$ for each edge $e \in E(G) - E(g)$.
Let $\Omega(g)$ denote the set of all possible worlds $g^{L}$ in which $g$ is embedded.
Obviously, if $E(g) \neq E(g')$, we must have $\Omega(g) \cap \Omega(g') = \emptyset$.
Meanwhile, for each possible world $g$, we easily have
\begin{equation*}
\sum_{g^{L} \in \Omega(g)} \Pr(g^{L}) = \Pr(g).
\end{equation*}
Therefore, for all possible worlds $g \in \G$, $\Omega(g)$ forms a division of the set of all possible worlds ${\G}^{L}$.

Notice that, we have
\begin{equation*}
I_G(u, v) = \sum_{g \in \G| u \text{ can reach } v \text{ in } g} \Pr(g).
\end{equation*}
For each possible world $g \in \G$, if $u$ can reach $v$ on $g$, there must exist a path from $u$ to $v$ on $g$.
At this time, for any possible world $g^{L} \in \Omega(g)$, it must embed at least one walk from $u$ to $v$.
On the other hand, for a possible world $g^{L} \in \Omega(g')$, if it does not embed any walk from $u$ to $v$, $u$ cannot reach $v$ on $g'$.
Thus, we have
\begin{equation*}
\begin{split}
& I_G(u, v) = \sum_{g \in \G| u \text{ can reach } v \text{ on } g} \sum_{g^{L} \in \Omega(g)} \Pr(g^{L}) \\
& = \sum_{g^{L} \in {\G}^{L}| \text{ there exists a walk from } u \text{ to } v \text{ on } g^{L}} \Pr(g^{L}) = \Pr(\bigvee_{i = 1}^{h_{uv}} W_i).
\end{split}
\end{equation*}
By the inclusion-exclusion principle, we can further expand the probability $\Pr(\bigvee_{i = 1}^{h_{uv}} W_i)$ as Eq.~(6). Thus, the lemma holds.
\end{proof}

\begin{figure*}[t]
\centering
\includegraphics[width= \textwidth]{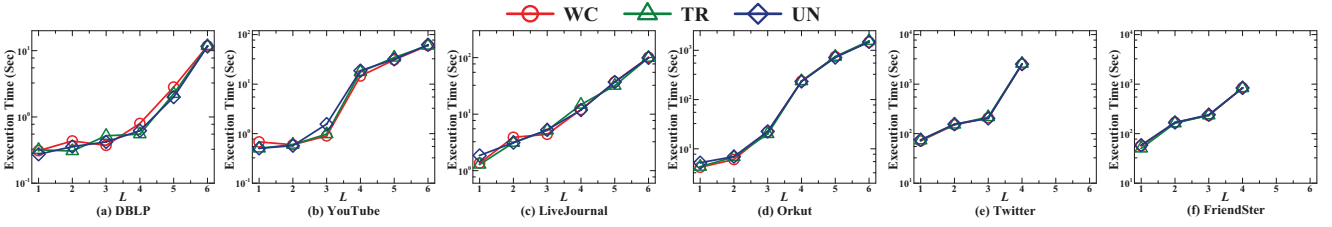}
%
\vspace{-1.5em}
\caption{Effects of Parameter \emph{L} on Execution Time.}
\label{Fig: Exp: TLL}
\end{figure*}

\begin{figure*}[t]
\centering
\includegraphics[width= \textwidth]{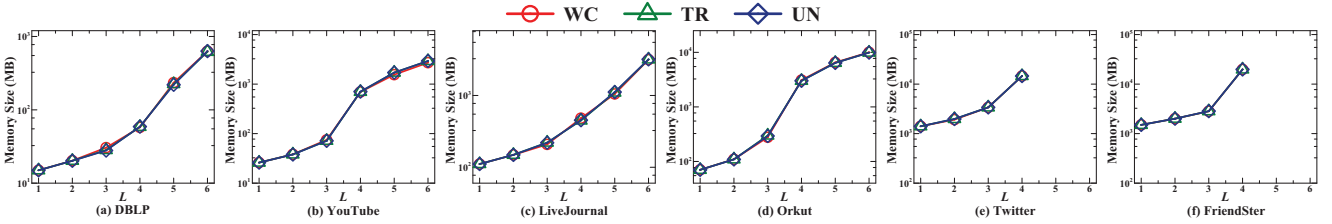}
%
\vspace{-1.5em}
\caption{Effects of Parameter \emph{L} on Memory Size.}
\label{Fig: Exp: MLL}
\end{figure*}

\begin{figure*}[t]
\centering
\includegraphics[width= \textwidth]{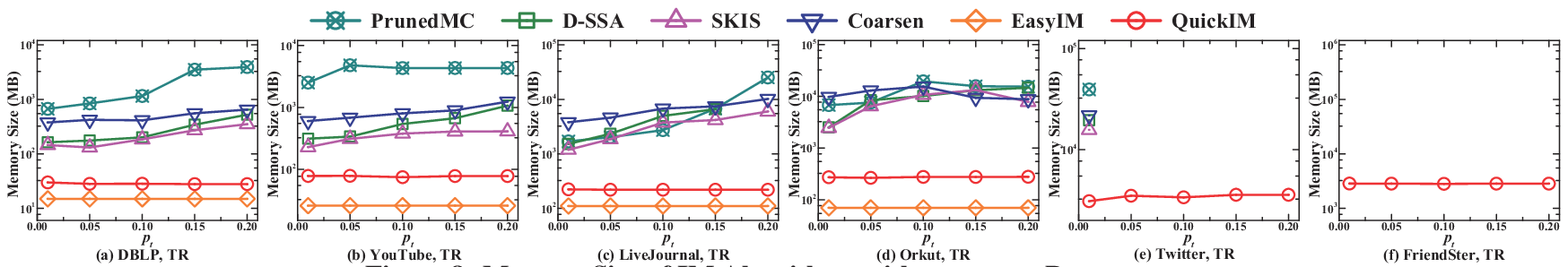}
%
\vspace{-1.5em}
\caption{Memory Size of IM Algorithms with respect to Parameter $p_{t}$.}
\label{Fig: Exp: PTM}
\end{figure*}

\begin{figure*}[t]
\centering
\includegraphics[width= \textwidth]{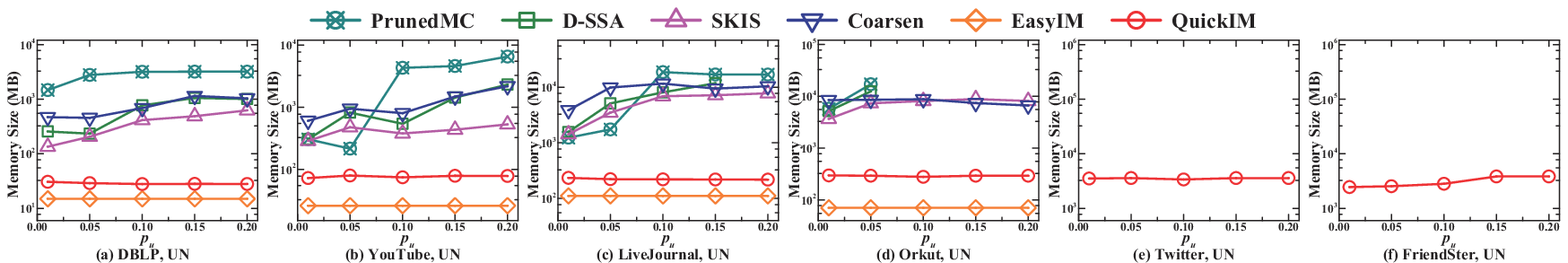}
%
\vspace{-1.5em}
\caption{Memory Size of IM Algorithms with respect to Parameter $p_{u}$.}
\label{Fig: Exp: PUM}
\end{figure*}

\subsection{Proof of Lemma~5}
\begin{lemma}
We have
\label{Lem: XRwIW}
$I_{G}(u, v) = \sum_{t=1}^{h_{uv}} X^{\langle t \rangle}_{G}(u, v)$, and $W_{G}(u, v) \break = \sum_{t=1}^{h_{uv}} t X^{\langle t \rangle}_{G}(u, v)$.
\end{lemma}
\begin{proof}
For $1 \leq i \leq h_{uv}$, let $\Upsilon_{i} \subseteq {\G}^{L}$ denote the set of all possible worlds $g^{L}$ in which there exist $i$ walks from $u$ to $v$.
We easily have $X^{\langle t \rangle}_{G}(u, v) = \sum_{g^{L} \in \Upsilon_{i}} \Pr(g^{L})$.
For any $i \neq j$, we have $\Upsilon_{i} \cap \Upsilon_{j} = \emptyset$.
For any possible world $g^{L}$ that embeds at least one walk from $u$ to $v$, there must exist $i' \in \mathbb{N}$ such that $g^{L} \in \Upsilon_{i'}$.
Thus, $\Upsilon_{1}, \Upsilon_{2}, \dots, \Upsilon_{h_{uv}}$ is a division of all possible worlds that embeds walks from $u$ to $v$.
By Lemma~4, we have
\begin{equation*}
I_{G}(u, v) = \Pr(\bigvee_{i = 1}^{h_{uv}} W_i) = \!\!\!\!\!\! \sum_{1 \leq i \leq h_{uv}} \sum_{g^{L} \in \Upsilon_{i}} \Pr(g^{L}) = \sum_{t=1}^{h_{uv}} X^{\langle t \rangle}_{G}(u, v).
\end{equation*}

For $W_{G}(u, v)$, we have
\begin{equation*}
W_G(u, v) = \sum_{i=1}^{h_{uv}} \Pr(W_{i}) = \sum_{i}^{h_{uv}} \sum_{g^L \in \G| g^L \text{ embeds walk} W_i} \Pr(g^L).
\end{equation*}
For each possible world $g^{L}$, if there exists $t$ walks on $g^{L}$ from $u$ to $v$, the probability $\Pr(g^L)$ will be counted $t$ times in $W_{G}(u, v)$
Therefore, we easily have
\begin{equation*}
W_G(u, v) = \sum_{t=1}^{h_{uv}} t \sum _{g^L \in \Upsilon_{t}} \Pr(g^L) =  \sum_{t=1}^{h_{uv}} t X^{\langle t \rangle}_{G}(u, v).
\end{equation*}
Thus, the lemma holds.
\end{proof}

\subsection{Proof of Lemma~6}

\begin{lemma}
\label{Lem: InfluInw}
Let $w$ be a vertex of $G$. Let $G_1$ be the graph obtained by removing all the incident edges of $w$ from $G$. Let $G_2$ be the graph obtained by removing all the out-edges of $w$ from $G$. For any vertex $u \neq w$, we have
$0 \leq |I_{G_{1}}(u) - I_{G_{2}}(u)| \leq 1$ and $|\widehat{I}_{G_1}(u)  - \widehat{I}_{G_2}(u) | \leq  p_{m}^{3} h_{uw} 2^{h_{uw}}$, where $h_{uw}$ is the number of walks from $u$ to $w$.
\end{lemma}
\begin{proof}
Since the vertex $w$ have no out-edge in both $G_1$ and $G_2$, there exist no path and walk from a vertex $u$ to other vertex $v$ via $w$.
Therefore, for each vertex $u \in V(G) $ and $u \neq w$, we have $I_{G_1}(u, v) = I_{G_2}(u, v)$ and $W_{G_1}(u, v) = W_{G_2}(u, v)$ for each vertex $v \neq w$.
On the graph $G_1$, since $w$ has no in-edge, we have $I_{G_1}(u, w) = 0$ and $W_{G_1}(u, w) = 0$ for each vertex $u \neq w$.
Thus, for each vertex $u \neq w$, $I_{G_{1}}(u) - I_{G_{2}}(u) = I_{G_2}(u, w)$.
By Eq.~\eqref{Eqn: IGSv}, we always have $0 \leq |I_{G_{1}}(u) - I_{G_{2}}(u)| \leq 1$.
Meanwhile, we have $|\widehat{I}_{G_1}(u)  - \widehat{I}_{G_2}(u)| =  I_{G_2}(u, w)$.
By Eq.~\eqref{Eqn: Epsiuv} and Eq.~\eqref{Eqn: IGuBound}, we always have $|\widehat{I}_{G_1}(u)  - \widehat{I}_{G_2}(u) | \leq  p_{m}^{3} h_{uw} 2^{h_{uw}}$.
\end{proof}

\section{Score Updating Procedures}

We present the details of the basic score updating method proposed in Section~\ref{Sec: ScrUpt-1}.

\myskip
\noindent{\textbf{\underline{Score Update Procedure.}}}
We present the \textsf{ScoreUpd-Basic} procedure to update scores.
The procedure takes as input the graph $G$, the maximum walk length $L$, the iteration number $t$, and the seed set $S$.
We assume that all vectors $\F^{(t)}_i$ and $\F^{(t)}$ have already been stored in the main memory.
First, we select the vertex $w^{(t)}$ with the maximum score (line~1) and prepare the matrix $\M^{(t)}$ (line~2).
Then, we invoke the \textsf{WalkPro} procedure to compute the column ${\A^{^{(t)}}}^j[*, w^{(t)}]$ for each $1 \leq j \leq L$ (line~3).
Next, we add $w^{(t)}$ to the seed set $S$ and remove all the out-edges of $w^{(t)}$ from $G$ (line~4).
We initialize $\F^{(t+1)}$ and $\F^{(t+1)}_{i}$ for all $2 \leq i \leq L$ to be $\0$ (line~5).
In fact, only the $w^{(t)}$-th element of $\F^{(t+1)}_1$ is different from that of $\F^{(t)}_1$, so we directly set $\F^{(t+1)}_1 = \F^{(t)}_1$
and set $\Delta \F^{(t)}_0$ and $\Delta \F^{(t)}_1$ to be $\0$ (line~6). After that, we compute $c^{(t)}_{0}$ and $c^{(t)}_{1}$ by Eq.~\eqref{Eqn: cxt} (line~7).

In the main loop (lines~8--13), for $2 \leq i \leq L$, we have already obtained $c^{(t)}_{j}$ for all $j \leq i-1$.
Therefore, we can update the estimated scores. Specifically, for each vertex $u \not\in S$,
we compute $\Delta \F^{(t)}_i[u]$ by Eq.~\eqref{Eqn: DeltaFiRec}, update $\F^{(t+1)}_i[u]$ to be $\F^{(t+1)}_i[u] + \Delta \F^{(t)}_i[u]$ and add $\F^{(t+1)}_i[u]$ to $\F^{(t+1)}[u]$ (lines~10--12).
After processing all vertices, we obtain the value of $c^{(t)}_{i}$ by Eq.~\eqref{Eqn: cxt} and store it in the main memory (line~13).

After the main loop, we remove all the in-edges of $w^{(t)}$ from $G$ to ensure that $w^{(t)}$ will not be accessed in the following iterations (line~14).
Finally, we store $\F^{(t+1)}$ and all $\F^{(t+1)}_i$ in the main memory for subsequent score updating (line~15).

\myskip
\noindent{\textbf{\underline{The \textsf{WalkPro} Procedure.}}}
${\A^{{(t)}}}^j[u, w^{(t)}]$ is the sum of probabilities of walks from $u$ to $w^{(t)}$ that are of length $j$.
To compute it, the \textsf{WalkPro} procedure performs a local traversal starting from $w^{(t)}$.
Initially, for each vertex $u \neq w$, we set ${\A^{{(t)}}}^j[u, w] = 0$  for $0 \leq j \leq L$ (line~1). Conceptually, we set ${\A^{{(t)}}}^{0}[w, w] = 1$ (line~2).
At the beginning, let $j = 0$. We iterate until $j = L-1$.
In iteration $j$ (lines~4--6), for each element ${\A^{{(t)}}}^j[u, w] \neq 0$, we fetch all in-neighbors $v$ of $u$ on $G$.
Obviously, the probability of all walks from $v$ to $w$ via $u$ is ${\A^{{(t)}}}^j[u, w]P_{G}(v, u)$, so we add it to ${\A^{{(t)}}}^{j+1}[v, w]$ (line~6).
After $j$ iterations, the probabilities ${\A^{{(t)}}}^j[u, w]$ for all vertices $u$ of $G$ and $1 \leq  j \leq L$ are computed and stored in the main memory (line~7).

\begin{figure}[!t]
    \centering
    \scriptsize
    \resizebox{\algswidth}{!}{
    \fbox{
    \parbox{\figwidth}{
    {
    \textbf{Procedure} \textsf{ScoreUpd-Basic}$(G, L, t, S)$
    \begin{algorithmic}[1]
    \STATE $w^{(t)} = \arg\max_{v \in V(G) - S} \F^{(t)}[v]$
    \STATE obtain the matrix $\M^{(t)}$
    \STATE invoke the procedure \textsf{WalkPro}$(G, L, w^{(t)})$
    \STATE add $w^{(t)}$ into the seed set $S$ and remove all out-going edges of $w^{(t)}$ from $G$
	\STATE $\F^{(t+1)} \gets \0$; $\F^{(t+1)}_{i} \gets \0$ for each $2 \geq i \geq L$
    \STATE  $\F^{(t+1)}_{1} \gets \F^{(t)}_{1}$; $\Delta \F^{(t)}_0, \Delta F^{(t)}_1 \gets \0$
    \STATE compute $c^{(0)}_{1}$ and $c^{(t)}_{1}$ by Eq.~\ref{Eqn: cxt}
	\FOR{$i \gets 2$ to $L$}
        \FOR{each vertex $u \in V(G) - S$}
            \STATE compute $\Delta \F^{(t)}_i[u]$ by Eq.~\ref{Eqn: DeltaFiRec}
            \STATE $\F^{(t+1)}_i[u] \gets \F^{(t)}_i[u] + \Delta \F^{(t)}_i[u]$
            \STATE $\F^{(t+1)}[u] \gets \F^{(t+1)}[u] + \F^{(t+1)}_i[u]$
        \ENDFOR
        \STATE compute $c^{(t)}_{i}$ by Eq.~\ref{Eqn: cxt}
	\ENDFOR
    \STATE remove all in-edges of $w^{(t)}$ from $G$
	\STATE store all vectors $\F^{(t+1)}$ and $\F^{(t+1)}_i$ for each $i$ in the main memory
    \end{algorithmic}
    }
    }}}
    \label{Fig: IncInf}
    \vspace{-1.5em}
\end{figure}

\begin{figure}[!t]
    \centering
    \scriptsize
    \resizebox{\algswidth}{!}{
    \fbox{
    \parbox{\figwidth}{
    {
    \textbf{Procedure} \textsf{WalkPro}$(G, L, w)$
    \begin{algorithmic}[1]
    \STATE ${\A^{{(t)}}}^j[u, w] \gets 0$ for each vertex $u \neq w$ and $0 \leq j \leq L$
	\STATE ${\A^{{(t)}}}^{0}[w, w] \gets 1$
	\FOR{$j \gets 0$ to $L-1$}
		\FOR{each element ${\A^{{(t)}}}^{j}[u, w] \neq 0$}
			\FOR{each vertex $v \in N_{G}^{I}(u)$}
				\STATE add ${\A^{{(t)}}}^{j}[u, w]P_{G}(v, u)$ onto ${\A^{{(t)}}}^{j+1}[v, w]$
			\ENDFOR
		\ENDFOR
	\ENDFOR
    \STATE store ${\A^{{(t)}}}^j[u, w]$ for all vertices $u$ of $G$ and $1 \leq j \leq L$ in the main memory
    \end{algorithmic}
    }
    }}}
    \label{Fig: ExtWPro}
    \vspace{-2em}
\end{figure}

\myskip
\noindent{\textbf{\underline{Complexity Analysis.}}}
It takes $O(Ln)$ time to compute $\Delta \F^{(t)}_i$ and update $\F^{(t)}_i$ according to Eq.~\eqref{Eqn: DeltaFiRec}.
The time to compute $c_x$ for each $1 \leq x \leq L$ is $O(d_{G}^{O}(w^{(t)}))$ by Eq.~\eqref{Eqn: cxt}.
Let $\Delta$ be the average degree of vertices in $G$. The expected time cost of \textsf{WalkPro} is $O(\Delta^{L})$.
Since each edge is traversed at most once in \textsf{WalkPro} for any $j$, the worst-case time cost of \textsf{WalkPro} is $O(Lm)$.
Since \textsf{ScoreUpd-Basic} iterates at most $L$ times, the expected and the worst-case time complexities of \textsf{ScoreUpd-Basic} are $O(L^{2}n + Ld_{G}^{O}(w) + \Delta^{L}) = O(L^{2}n + \Delta^{L})$ and $O(L^{2}n + Lm)$, respectively.
\textsf{ScoreUpd-Basic} stores $\Delta \F^{(t)}_i$ and ${\A^{{(t)}}}^j[u, w^{(t)}]$ for $1 \leq j \leq L$.
Thus, the expected space complexity of \textsf{ScoreUpd-Basic} is $O(\Delta^{L})$. In the worst-case, ${\A^{{(t)}}}^j[u, w^{(t)}]$ has $n$ elements for each $1 \leq j \leq L$, so the worst-case space complexity of \textsf{ScoreUpd-Basic} is $O(Ln)$.

When we apply \textsf{ScoreUpd-Basic} in all of the $k$ iterations. The total expected and the total wort-case time costs are indeed $(k(L^{2}n + \Delta^{L}))$ and $O(k(L^{2}n + Lm))$, respectively.
Meanwhile, the expected space complexity is still $O(\Delta^{L})$ since we do not need to reserve any data in the main memory after the procedure.

\section{Additional Experiments}

We present some additional experimental results in this section.

\myskip
\noindent{\textbf{\underline{Execution Time v.s.~Parameter $L$.}}}
The execution time of \textsf{QuickIM} w.r.t.~parameter $L$ is shown in Fig.~\ref{Fig: Exp: TLL}.
We have the following observations:
1) The execution time of \textsf{QuickIM} grows significantly as $L$ becomes larger.
This is because the time complexity of \textsf{QuickIM} is $O(Lm + kLn + kL^{2}n' + k\Delta^{L})$.
The item $k \Delta^L$ grows exponentially to $L$, so the time cost grows fast.
2) In different probability assigning models, the execution time of \textsf{QuickIM} is very close.
The reasons have been explained clearly in Section~\ref{Sec: PerEva-2}.
Once again, this verifies the robustness of \textsf{QuickIM} in terms of time efficiency.

 \myskip
\noindent{\textbf{\underline{Memory Overhead v.s.~Parameter $L$.}}}
The memory overhead of \textsf{QuickIM} with respect to parameter $L$ is shown in Fig.~\ref{Fig: Exp: MLL}.
We have the following observations:
1) The memory overhead of \textsf{QuickIM} also grows significantly as $L$ gets larger.
This is because the space complexity of \textsf{QuickIM} is $O(Ln + k\Delta^{L})$.
The item $k \Delta^L$ grows exponentially to $L$, so the space cost also grows fast.
2) In different probability assigning models, the memory overhead of \textsf{QuickIM} is very close.
The reasons have been explained clearly in Section~\ref{Sec: PerEva-2}.
Once again, this verifies the robustness of \textsf{QuickIM} in terms of memory efficiency.

Notice that, although the time and the space costs grow significantly with respect to $L$, $L = 3$ is sufficient to
produce high quality results as we verified in Section~\ref{Sec: PerEva-2}.

\myskip
\noindent{\textbf{\underline{Memory Overhead in Robustness Evaluation.}}}
In this experiment, we further examine the robustness of the IM algorithms against influence probabilities.
Let $k = 100$. We vary the parameter $p_{t}$ in the \textsl{TR} model and the parameter $p_{u}$ in the \textsl{UN} model from 0.01 to 0.2.
The memory overheads of the IM algorithms with respect to $p_{t}$ and $p_{u}$ are illustrated in Fig.~\ref{Fig: Exp: PTM} and Fig.~\ref{Fig: Exp: PUM}, respectively.
We find that the memory overhead of \textsf{QuickIM} is independent of $p_{t}$ and $p_{u}$.
However, the memory overheads of all the sampling-based algorithms \textsf{PrunedMC}, \textsf{D-SSA}, \textsf{SKIS} and \textsf{Coarsen} grow exponentially as $p_{t}$ or $p_{u}$ gets larger.
The reasons have been clarified clearly in Section~\ref{Sec: PerEva-2}. The memory overhead of \textsf{EasyIM} is also independent of $p_{t}$ and $p_{u}$.

\end{document}